\documentclass[11pt]{amsart}
\usepackage{amssymb,amsmath,mathtools}
\usepackage{enumitem}
\usepackage{xcolor}
\usepackage{caption}
\usepackage[centering]{geometry}
\usepackage[utf8]{inputenc}
\usepackage{graphicx}
\usepackage{float}
\usepackage{tabu}
\usepackage{makecell}
\usepackage{booktabs}
\usepackage{multirow}

\newtheorem{theorem}{Theorem}[section]
\newtheorem{proposition}[theorem]{Proposition}
\newtheorem{lemma}[theorem]{Lemma}
\newtheorem{corollary}[theorem]{Corollary}

\theoremstyle{definition}
\newtheorem{definition}[theorem]{Definition}

\newtheorem{remark}[theorem]{Remark}

\newcommand{\eins}{\text{\ensuremath{1\hspace*{-0.9ex}1}}}
\numberwithin{equation}{section}

\newcommand{\R}{\mathbb{R}}

\newcommand{\E}{\mathbb{E}}
\renewcommand{\P}{\mathbb{P}}
\newcommand{\C}{\mathrm{\mathbb{C}ov}}
\newcommand{\V}{\mathrm{\mathbb{V}ar}}

\newcommand{\ES}{\mathrm{ES}}
\newcommand{\VAR}{\mathrm{VaR}}

\numberwithin{equation}{section}

\begin{document}

\title{Optimal nonparametric estimation of the expected shortfall risk}

\author{Daniel Bartl \and Stephan Eckstein}
\address{National University of Singapore, Department of Mathematics, Department of Statistics and Data Science}
\email{bartld@nus.edu.sg}
\address{University of T\"{u}bingen, department of mathematics, Auf der Morgenstelle 10, 72076 Tübingen, Germany}
\email{stephan.eckstein@uni-tuebingen.de}
\date{\today}

\begin{abstract}
We address the problem of estimating the expected shortfall risk of a financial loss using a finite number of i.i.d.~data.
It is well known that the classical plug-in estimator suffers from poor statistical performance when faced with (heavy-tailed) distributions that are commonly used in financial contexts.
Further, it lacks robustness, as the modification of  even a single data point can cause a significant distortion.
We propose a novel procedure for the estimation of the expected shortfall and prove that it recovers the best possible statistical properties (dictated by the central limit theorem) under minimal assumptions and for all finite sample sizes.
Further, this estimator is adversarially robust: even if a (small) proportion of the data is maliciously modified, it still provides optimal estimates of the true expected shortfall risk.
We demonstrate that our estimator outperforms the classical plug-in estimator through a variety of numerical experiments across a range of standard loss distributions.
\end{abstract}

\maketitle
\setcounter{equation}{0}
\setcounter{tocdepth}{2}


\section{Introduction}

A central task in risk management involves accurately quantifying the level of risk associated with a given financial position.
A rigorous framework addressing this question is provided by the theory of risk measures (see, e.g., \cite{Artzner1999, dhaene2006risk, mcneil2015quantitative}), and the most widely employed method in practice involves using the \emph{expected shortfall} risk measure (cf.~\cite{embrechts2014academic,Rock-Ury00,rockafellar2002conditional}).
For a random financial loss $X$,  the expected shortfall at the level $\alpha\in(0,\frac{1}{2})$ is defined as
\begin{align*}
\mathrm{ES}_\alpha(X)
&:= \frac{1}{\alpha} \int_{1-\alpha}^1 \VAR_u (X)\,du,
\end{align*}
where $\VAR_u(X):=\inf\{ t \in\mathbb{R} : \P(X\leq t)\geq u \}$ is the value at risk at the level $u \in [1-\alpha, 1)$ of the loss $X$ (cf.~\cite{mcneil2015quantitative, duffie1997overview}).

However, calculating $\mathrm{ES}_\alpha(X)$ requires knowledge of the distribution function $F_X(t) = \P(X \leq t)$ of $X$, which is rarely known precisely. 
Instead, one commonly relies on \emph{historical data}, denoted as $X_1, \dots, X_N$. 
While these data may typically be subject to some mild correlation, in this article, we adopt the standard approach found in many papers addressing the estimation of the expected shortfall (see, e.g., \cite{acerbi2007coherent,brazauskas2008estimating,cont2010robustness,hill2015expected,jones2003empirical,yamai2002comparative}) and assume the idealized scenario where the data are independent and identically distributed.
While this simplifying assumption is used throughout most of the paper and in the following discussion, we emphasize that the framework developed in this paper is sufficiently flexible to accommodate weak dependence structures as well, one of which is treated in Section \ref{sec:non.iid}.

\subsection{The plug-in estimator}

The arguably most natural estimator for $\ES_\alpha(X)$ using the data $X_1,\dots,X_N$ is the \emph{plug-in estimator}.
Indeed, noting that $\ES_\alpha(X)$ is a quantity that only depends on the distribution function $F_X$ of $X$, the plug-in estimator $\widehat{T}_N$ consists in taking the  empirical distribution function $\widehat{F}_N(t):= \frac{1}{N} \sum_{i=1}^N \eins_{(-\infty,t]}(X_i)$ as a proxy for $F_X$ and computing 
\begin{align*} 
\widehat{T}_N
&:= \frac{1}{\alpha}\int_{1-\alpha}^1 \VAR_u\left( \widehat{F}_N^{-1} \right)\,du,
\end{align*}
where $\VAR_u( \widehat{F}_N )=\widehat{F}_N^{-1}(u):=\inf\{t: \widehat{F}_N(t)\geq u\}$.

The central question of interest in the present context  is then the extent to which the estimator $\widehat{T}_N$ approximates $\ES_\alpha(X)$.
More accurately, since $\widehat{T}_N$ is a random quantity, one is interested in obtaining bounds on the error $\widehat{T}_N- \ES_\alpha(X)$ that hold with  \emph{high confidence}.
As a first step towards that goal, it often turns out helpful to analyse the \emph{asymptotic} behavior as $N\to\infty$. 
In fact, the latter has been studied in detail and is well understood by now.
To formulate the results, denote by 
\begin{align*}
\sigma^2_{\ES_\alpha(X)}
:=\frac{1}{\alpha^2} \int_{\VAR_{1-\alpha}(X)}^\infty \int_{\VAR_{1-\alpha}(X)}^\infty  F_X(\min\{t,s\}) - F_X(t)F_X(s) \,dtds
\end{align*}
the variance associated to the estimation problem of the expected shortfall.
Note that $\sigma^2_{\ES_\alpha(X)}$ is always well-defined taking values in $[0,\infty]$, and it is finite whenever the positive part of $X$ is square integrable.
With this notation set in place, the following central limit theorem for the convergence of $\widehat{T}_N$ towards $\ES_{\alpha}(X)$ is known.

\begin{theorem}[\cite{brazauskas2008estimating,pflug2010asymptotic}]
\label{thm:clt}
	Assume that $\sigma^2_{\ES_\alpha(X)}$ is finite  and that  $u\mapsto \VAR_{u}(X)$ is continuous in $1- \alpha$.
	Then, as $N\to\infty$, we have the weak convergence
\begin{align*}
\sqrt{N}\left( \widehat{T}_N-  \ES_\alpha(X)  \right)
\to \mathcal{N}\left( 0, \sigma^2_{\ES_\alpha(X)} \right).
\end{align*}
\end{theorem}

Here, $\mathcal{N}( 0,\sigma^2)$ denotes the normal distribution with mean zero and variance $\sigma^2$.
While it is  obvious that $\sigma^2_{\ES_\alpha(X)}<\infty$ is a necessary assumption in Theorem \ref{thm:clt}, it might be less obvious  whether  the assumption regarding the continuity of   $u\mapsto \VAR_{u}(X)$ is necessary.
As it happens, this assumption is indeed necessary, since without continuity $u\mapsto \VAR_{u}(X)$ in general 
\[\sqrt{N}\left( \widehat{T}_N-  \ES_\alpha(X)  \right)
\not\to \mathcal{N}\left( 0, \sigma^2_{\ES_\alpha(X)} \right).\]
We refer to \cite[Example 3.4]{pflug2010asymptotic} as well as Remark \ref{rem:lipschitz} and  Section \ref{sec:continuity.inverse} for more details.

\subsection{Finite sample performance of the plug-in estimator}

What is considerably more relevant than the central-limit-theorem–type result of Theorem~\ref{thm:clt}---which describes the asymptotic distribution as the sample size $N \to \infty$---is the finite-sample performance: how well does $\widehat{T}_N$ approximate $\ES_\alpha(X)$ for finite $N$? 
Mathematically, given a confidence level $1 - \delta \in (0,1)$, the goal is to find a number $r=r(\delta,N,X)$ such that the following holds:
\begin{align*}
\P\left( \left| \widehat{T}_N-  \ES_\alpha(X)  \right| \leq r \right) \geq 1-\delta, 
\end{align*}
where $\P$ is the probability taken over the i.i.d.\ samples $X_1,\dots,X_N$.
Naturally, $r$ should be as small as possible.

Upon initial examination, it may seem that Theorem \ref{thm:clt} offers a satisfactory answer to that question. 
Indeed, an overly optimistic interpretation of the normal approximation presented in Theorem \ref{thm:clt} consists in assuming that the normal approximation holds as an equality for finite $N$, which leads to 
\begin{align}
\label{eq:exp.rate.confidence}
\P\left( \left|   \widehat{T}_N  - \ES_\alpha(X)\right| 
\leq \sigma_{\ES_\alpha(X)}  \sqrt{ \frac{2}{N}  \log \left(\frac{2}{\delta} \right) }  \right) \approx 1-\delta .
\end{align}
In particular, the dependence of $r$ on $\delta$ is only logarithmically, meaning that one can obtain a very high confidence (i.e.\ $\delta$ very close to zero) without increasing the error by too much.
Note that an equivalent reformulation of \eqref{eq:exp.rate.confidence}---which we mostly adopt in this paper---is in terms of a Gaussian tail-estimate on the confidence:
\begin{align}
\label{eq:exp.rate}
\P\left( \left| \widehat{T}_N
- \ES_\alpha(X)\right| \geq \varepsilon \sigma_{\ES_\alpha(X)} \right)
	\approx 2\exp\left(- \frac{1}{2}N\varepsilon^2 \right).
	\end{align}

However, it is crucial to emphasize 
that the interpretation of Theorem \ref{thm:clt} in  \eqref{eq:exp.rate.confidence} and \eqref{eq:exp.rate}
is in general \emph{incorrect}.

The foremost reason is that an exponential rate as in  \eqref{eq:exp.rate} is simply false unless $X$ is very light-tailed, in the sense that it satisfies a Gaussian-like tail decay. 
Indeed, while it can be established that the estimate in \eqref{eq:exp.rate} is the best that one can hope for (see Section \ref{sec:no.estimtor.better.than.clt} for more details), the actual statistical behavior is significantly worse for distributions typically considered in finance and insurance.

 To illustrate this point, let us consider the case where $X$ follows a Pareto distribution with a scale parameter $\lambda>2$ (ensuring a finite variance). 
This distribution class is known for its heavy tails and is commonly used to model losses in insurance (see, e.g., \cite{beirlant1992modeling, wuthrich2023statistical}). 
Focusing on the dependence on $N$ for simplicity, we shall show in  Section \ref{sec:Pareto} that 
\begin{align}
	\label{eq:intro.paret}
	\P\left( \left| \widehat{T}_N
- \ES_\alpha(X)\right| 
\geq \varepsilon \sigma_{\ES_\alpha(X)} \right)
	\geq  
	\frac{C_\varepsilon}{ N^{\lambda-1}},
	\end{align}
which is far from the behavior suggested in \eqref{eq:exp.rate}.
	In fact, for values of $\lambda$ close to $2$, the right-hand side in \eqref{eq:intro.paret} scales almost as $\frac{1}{N}$ --- in other words, linearly as opposed to exponentially as in \eqref{eq:exp.rate}. 
	In this regard, Section \ref{sec:minimax} establishes that for more general classes of distributions, the linear rate 
	\[\P\left( \left| \widehat{T}_N
	- \ES_\alpha(X)\right| 
	\geq \varepsilon \sigma_{\ES_\alpha(X)} \right) \gtrsim \frac{1}{N\varepsilon^2}\] 
	is the best one can hope for.
    Equivalently, this means that 
    \[\P\left( | \widehat{T}_N	- \ES_\alpha(X)|  \lesssim \sigma_{\ES_\alpha(X)} \sqrt{\frac{1}{\delta N}}\right) \geq 1-\delta\]
    is the best dependence of the error on $\delta$ that one might hope for.
    For small $\delta$, $\sqrt{1/\delta N}$ is significantly larger than $\sqrt{\log(2/\delta) /N}$.

	An additional reason why \eqref{eq:exp.rate} is false stems from a distinct aspect, one that remains unaffected by the assumed tail behavior of $X$. 
	Indeed, Theorem \ref{thm:clt} states that if $u\mapsto \VAR_{u}(X)$ is continuous, then the typical behavior of the error is
	\begin{align}
	\label{eq:no.second.oder}
	\left| \widehat{T}_N
- \ES_\alpha(X)\right| 
= \mathcal{O}\left( \frac{\sigma_{\ES_\alpha(X)}}{\sqrt N} \right) 
	\end{align}
for \emph{sufficiently large} $N$, but it does not provide a characterization of when $N$ falls within this regime. 
It turns out that such a characterization is intricately tied to the quantitative continuity of the function $u\mapsto \VAR_{u}(X)$. 
Specifically, assuming $L$-Lipschitz continuity of the latter function (and, say, that $|X|\leq 1$ to alleviate concerns related to the tail behavior of $X$), we show that the typical behavior of the error has an additional `second order' term: 
	\begin{align}
	\label{eq:second.order}
	\E\left[ \left| \widehat{T}_N
- \ES_\alpha(X)\right|  \right]
 \gtrsim   \frac{\sigma_{\ES_\alpha(X)}}{\sqrt N}  +  \frac{L}{N} 
	\end{align}
	for all $N$, see Section \ref{sec:continuity.inverse} for more details and in particular Proposition \ref{prop:inverse.continuity} for a precise minimax result in this regard in particular showing that \eqref{eq:second.order} is also true with constant probability instead of just in-expectation.
	This shows that the regime suggested in \eqref{eq:no.second.oder} is the right one \emph{only} for  $N \gtrsim L^2/\sigma^2_{\ES_\alpha(X)}$.
	
	\begin{remark}
	\label{rem:lipschitz}
The reader familiar with quantile estimation will recognize the condition imposed on 
\(u \mapsto \VAR_{u}(X)\).
It is well known that continuity of \(u \mapsto \VAR_{u}(X)\) at \(1 - \alpha\) is necessary for the asymptotic consistency of the plug-in quantile estimator (and, in fact, of any estimator) whereas a non-asymptotic setting one typically requires Lipschitz continuity of \(u \mapsto \VAR_{u}(X)\) at \(1 - \alpha\), see, e.g., \cite{bahadur1966note,Boucheron2012}.

Note that if \(X\) admits a strictly positive density \(f_X\), then 
\[
\frac{d}{du}\VAR_{u}(X) = \frac{1}{f_X(F_X^{-1}(u))},
\]
so that Lipschitz continuity is essentially equivalent to the standard assumption in non-asymptotic quantile estimation, namely that 
\(f_X(F_X^{-1}(u)) \geq \frac{1}{L}\) for \(u\) close to \(1 - \alpha\).

Since \(\ES_\alpha\) represents the conditional mean beyond the level \(1 - \alpha\), 
it is natural that a quantitative continuity condition on the quantile function is required in the non-asymptotic estimation regime. 
To the best of our knowledge, however, this fundamental observation has not been explicitly stated or analyzed in previous work.
\end{remark}

\noindent
{\bf In summary:} the most optimistic estimate for estimating the expected shortfall risk that one can aspire to achieve  is
\begin{align}
\label{eq:best.possible}
\P\left( \left| \widehat{T}_N
- \ES_\alpha(X)\right| 
\gtrsim   \varepsilon \sigma_{\ES_\alpha(X)} + \varepsilon^2 L   \right)
	\leq 2 \exp\left(- \frac{1}{2}N\varepsilon^2 \right).
	\end{align}
Crucially, this estimate can only be true if $X$ has very light tails, comparable to a Gaussian distribution.
In more general settings, the performance of $\widehat{T}_N$ is typically significantly worse.

\subsection{An improved estimator}

Given that the tails of distributions used in finance and insurance are usually significantly heavier than those of a Gaussian distribution, the poor statistical performance of the plug-in estimator was addressed by many researchers.
A natural course of action going back to the early days of robust statistics  is to explore \emph{alternative} estimators --- i.e.\ functions $\widehat{R}_N\colon \R^N\to\R$ that assign to the data $(X_i)_{i=1}^N$ an estimate for $\ES_\alpha(X)$, and we refer to \cite{chen2008nonparametric, hill2015expected,necir2010estimating,yamai2005value,yamai2002comparative} for several works in this direction and the survey \cite{nadarajah2014estimation}.
Broadly speaking, the proposed estimators were usually demonstrated to enhance performance in statistical experiments, and/or their asymptotic statistical behavior was studied (with results similar to those in Theorem \ref{thm:clt}).
However, to the best of our knowledge, the statistical performance for finite samples has not been addressed so far --- in particular the presence of the second order term in \eqref{eq:second.order} remained unnoticed.

In this article, we construct an estimator that exhibits the  \emph{optimal} statistical behavior: 
irrespective of the tail-decay of $X$, if  $u\mapsto\VAR_{u}(X)$ is $L$-Lipschitz and $\sigma_{\ES_\alpha(X)}^2$ is finite, then the optimal estimate \eqref{eq:best.possible} is true once we replace $\widehat{T}_N$ by our proposed estimator.

\vspace{0.5em}
We start with the formal definition of the estimator and explain the intuition behind it in Remark \ref{rem:bias}. 
To that end, for $x>0$, denote by $\lfloor x \rfloor $ the largest integer smaller than $x$ and by  $\lceil x \rceil $ the smallest integer larger than $x$. 

\begin{definition}\label{def:proposed}
	Set $m:=\lceil \frac{11}{\varepsilon^2} \rceil$, $n:=\lfloor \frac{N}{m} \rfloor$ and
	\[I_j:=\{(j-1)m + 1,\dots, j m\}, ~~~ \text{ for } j = 1, \dots, n,\]
	so that $nm\leq N$ and $(I_j)_{j=1}^n$ forms a disjoint partition of $\{1,\dots, mn\}$. 
	For each  $j$, denote by $\widehat{T}_{I_j}$ the plug-in estimator using only the sample $(X_i)_{i\in I_j}$, i.e.,
	\[\widehat{T}_{I_j} 
	:= \frac{1}{\alpha} \int_{1-\alpha}^1 \VAR_u \left( \widehat{F}_{I_j} \right) \,du, ~~~ \text{where }~~~ \widehat{F}_{I_j}(t):=\frac{1}{m}\sum_{i\in I_j} \eins_{(-\infty,t]}(X_i). \]
	Denote by $\beta \mapsto \hat{Q}(\beta)$ the linearly interpolated empirical quantile function of $(\widehat{T}_{(j)})_{j=1}^n$.\footnote{
	Following \cite[Definition 7]{Hyndman-1996}, denote by  $(\widehat{T}_{(j)})_{j=1}^n$ the order statistics of  $(\widehat{T}_{I_j})_{j=1}^n$, set $\hat{Q}(0) = \widehat{T}_{(1)}, \hat{Q}(1) = \widehat{T}_{(n)}$,
	\[
	\hat{Q}\left(\frac{j-1}{n-1}\right) = \widehat{T}_{(j)}, ~ \text{ for } j=2, \dots, n-1,
	\]
	and let $\hat{Q}$ be the function that linearly interpolates between these values.
	}
	Then, for some hyperparameters $0.35 \leq \beta_1 \leq \beta_2 \leq 0.65$, we set
	\[
	\widehat{S}_N := \min\left\{ \max\left\{ \widehat{T}_N, \hat{Q}(\beta_1) \right\}, \hat{Q}(\beta_2) \right\}.
	\]
\end{definition}

To unpack Definition \ref{def:proposed}, we restate the estimation procedure more algorithmically:
	\begin{enumerate}
		\item Hyperparameters: Block size $m \in \mathbb{N}$ and quantile levels $0.35 \leq \beta_1 \leq \beta_2 \leq 0.65$.
		\item Input: Data $X_1, \dots, X_N$
        \item Step 1: Calculate the plug-in estimator $\widehat{T}_N$ on the whole dataset
        \item Step 2: Split the data into $n \approx \frac{N}{m}$ many disjoint blocks $I_1, \ldots, I_n$ of size $m$ each and calculate the block-wise plug-in estimators $\widehat{T}_{I_1}, \ldots, \widehat{T}_{I_n}$
        \item Step 3: Get the $\beta_1$- and $\beta_2$-quantiles $\hat{Q}(\beta_1)$ and $\hat{Q}(\beta_2)$ of $\widehat{T}_{I_1}, \ldots, \widehat{T}_{I_n}$
        \item Output $\widehat{S}_N$: Return $\widehat{T}_N$ truncated to the interval $[\hat{Q}(\beta_1), \hat{Q}(\beta_2)]$
	\end{enumerate}

We will explain the reasoning behind the estimator $\widehat{S}_N$ and discuss the importance of the hyperparameters $\beta_1$ and $\beta_2$ in Remark \ref{rem:bias}.
Before doing so, let us present our first main result pertaining to the statistical behaviour of $\widehat{S}_N$.

\begin{theorem}
\label{thm:ES}
	There are absolute constants $c_0,c_1,c_2, c_3>0$ such that the following holds.
	Assume that $\sigma^2_{\ES_\alpha(X)}$ is finite and that $u\mapsto \VAR_u(X)$ is $L$-Lipschitz continuous on $[1-2\alpha , 1- \frac{1}{2}\alpha]$.
	Then, for every $\varepsilon< c_0\sqrt\alpha$ and $N \geq c_1 \varepsilon^{-2}$, 
	\[ \P\left( \left| \widehat{S}_N-\mathrm{ES}_\alpha(X) \right| 
	\geq   \varepsilon\sigma_{\ES_\alpha(X)}  + c_2 L \varepsilon^2   \right)
	 \leq \exp\left(-c_3 N\varepsilon^2\right).\]
\end{theorem}

In particular, the estimator $\widehat{S}_N$ achieves the optimal statistical behavior.
Indeed, we have previously identified that no estimator can have a better performance regarding confidence intervals as those  given in \eqref{eq:best.possible} (and, crucially, standard estimators such as the plug-in estimators do not achieves this performance).
By Theorem \ref{thm:ES},  the estimator $\widehat{S}_N$ achieves those bounds  under minimal and necessary assumptions on $X$, perhaps up to the specific choice of the multiplicative constants $c_0,\dots,c_3$.

\begin{remark}
\label{rem:constants}
	Our proof shows that one may choose $c_0=\frac{1}{6}$, $c_1=140$, $c_2=10$, and $c_3=\frac{1}{700}$ though we believe that it is possible to obtain considerably better estimates for these constants.
	Moreover, in \mbox{Definition \ref{def:proposed}}, the constant $11$ when setting $m:=\lceil \frac{11}{\varepsilon^2} \rceil$ should be regarded suggestively: a different constant would change the statement of Theorem \ref{thm:ES} only via the potential need for different multiplicative constants. 
\end{remark}

\begin{remark}
\label{rem:bias}
Let us shortly discuss the intuition behind the proposed estimator.
While Theorem \ref{thm:ES} is valid for all choices of $0.35\leq \beta_1\leq \beta_2\leq 0.65$, our numerical experiments reveal that these values significantly affect the empirical \emph{bias} of $\widehat{S}_N$. 
This happens because the block estimators $\widehat{T}_{I_j}$ tend to be right-skewed. 
Thus, setting e.g.\ $\beta_1=\beta_2=0.5$ (resulting in $\widehat{S}_N$ being the median of blocks estimator, see Section \ref{subsec:relatedliterature} for a discussion) often leads to a negative bias. 
Conversely, choosing $\beta_1=\beta_2>0.5$ would induce a positive bias in situations where the central limit approximation for $\widehat{T}_{I_j}$ takes its effects (i.e.\ when $m$ is large).
Thus, the optimal choice of $\beta_1,\beta_2$ should depend on whether one believes the underlying distribution to be skewed and on the size of $m$. 
We present some numerical examples in Section \ref{sec:numerics} which showcase the impact of the choice of $\beta_1$ and $\beta_2$.

For choices $\beta_1<0.5<\beta_2$ the estimator $\widehat{S}_N$ can be interpreted in terms of adaptive truncation: if the data $(X_i)_{i=1}^N$ has no outliers, the plug-in estimator $\widehat{T}_N$ is likely to fall between $\widehat{Q}(\beta_1)$ and $\widehat{Q}(\beta_2)$; thus $\widehat{S}_N=\widehat{T}_N$ and since the latter averages all data points, it should have a small bias.
Conversely, if $\widehat{T}_N$ lies outside that range, one suspects the presence of outliers and our estimator truncates $\widehat{T}_N$ at $\widehat{Q}(\beta_1)$ and $\widehat{Q}(\beta_2)$.
In this regard, we can also understand the importance of the block size $m$: If $m$ is smaller, the variance of the estimates among each block is larger, but there are more blocks to obtain high confidence in the obtained quantiles. This means that with smaller $m$, the obtained truncation interval $[\widehat{Q}(\beta_1), \widehat{Q}(\beta_2)]$ is usually larger and truncation gives a high level of confidence against larger deviations $\varepsilon$. Increasing the size $m$ makes the truncation interval smaller and thus can have an effect already at smaller levels of deviations $\varepsilon$, although this is then based on a smaller number of blocks and thus confidence in these intervals lowers compared to smaller $m$.
\end{remark}

\begin{figure}
	\begin{minipage}{0.333\textwidth}
		\includegraphics[width=1\textwidth]{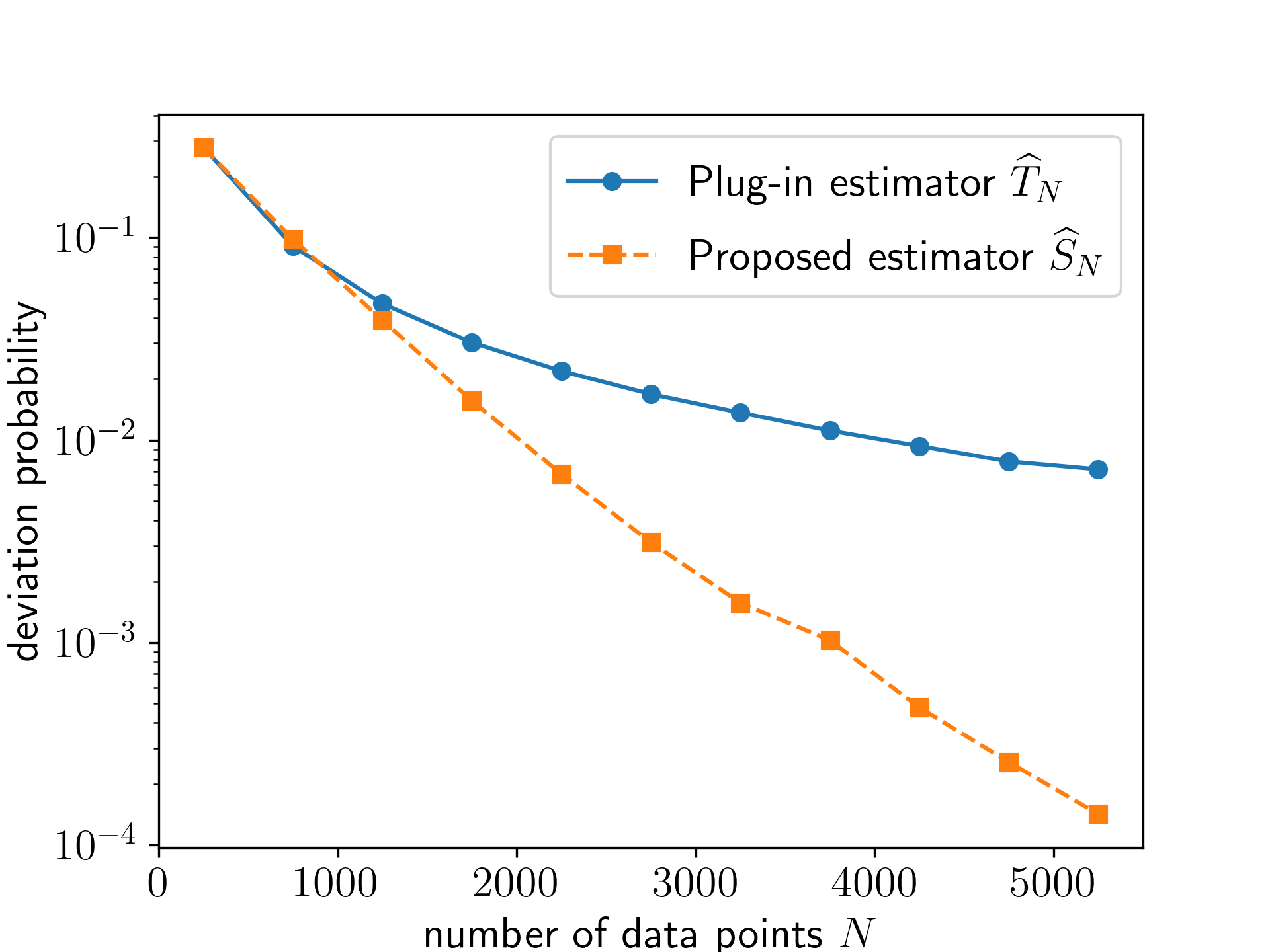}
	\end{minipage}%
	\begin{minipage}{0.333\textwidth}
		\includegraphics[width=1\textwidth]{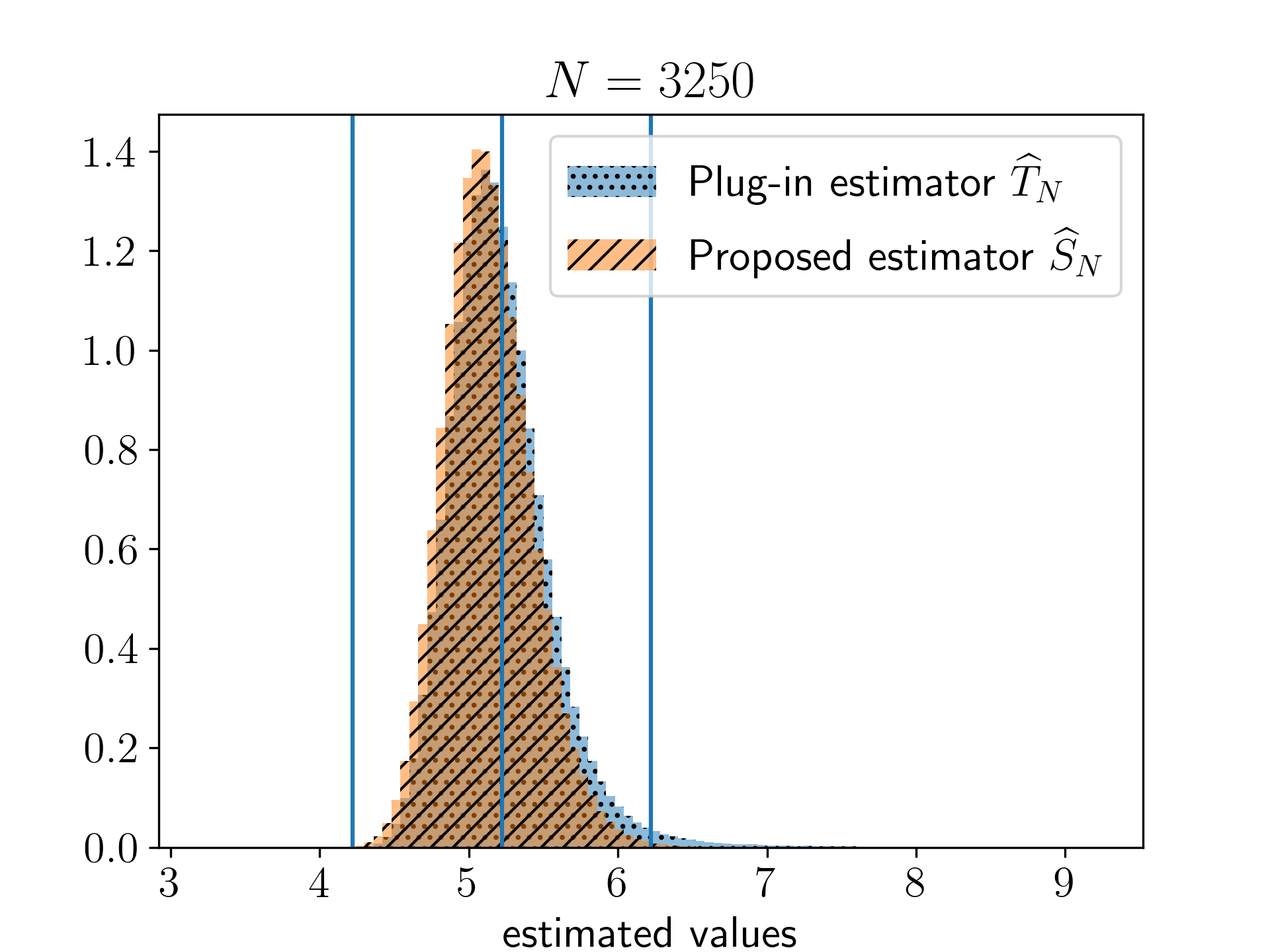}
	\end{minipage}%
	\begin{minipage}{0.333\textwidth}
		\includegraphics[width=1\textwidth]{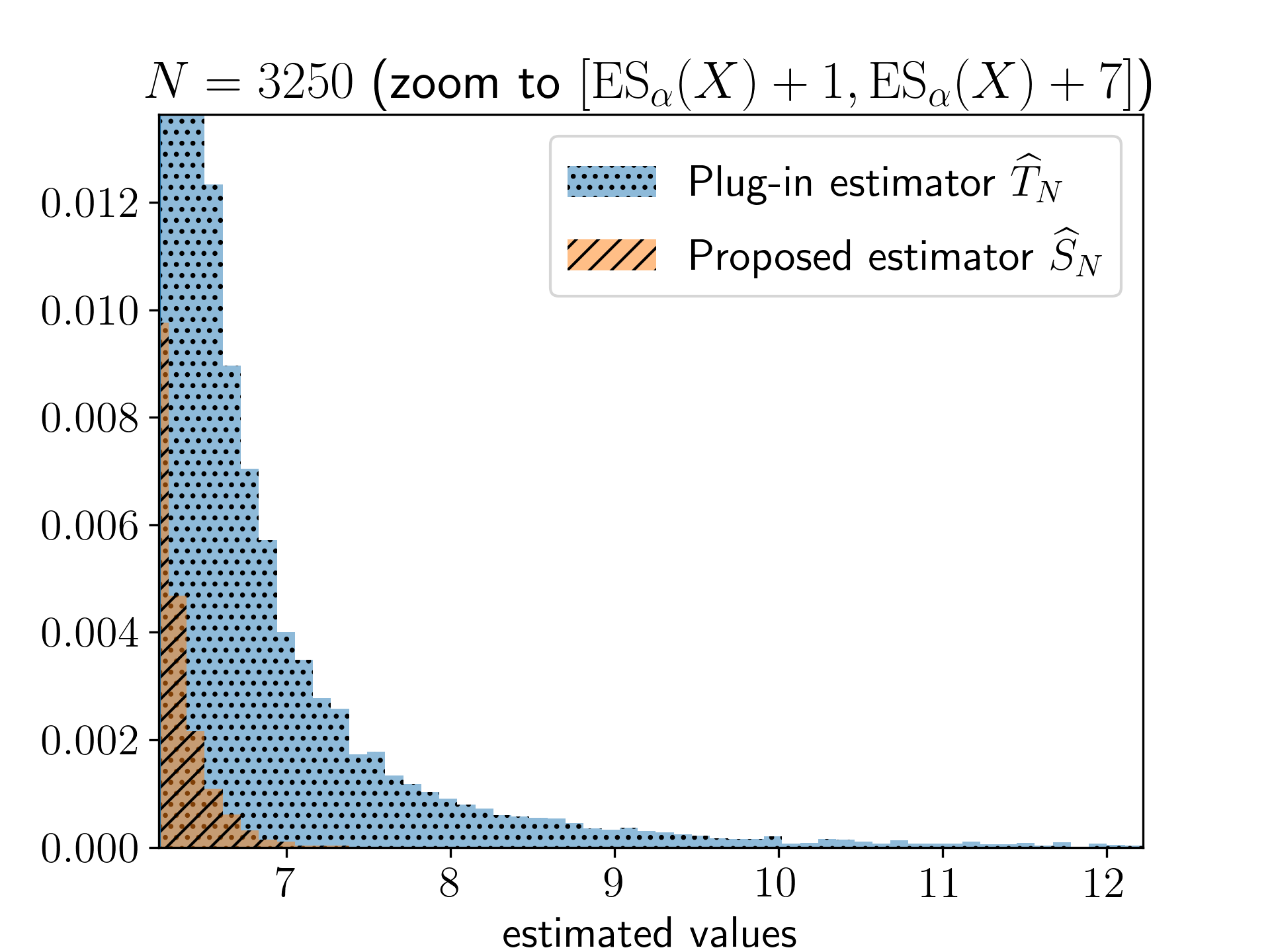}
	\end{minipage}
	\caption{The left image depicts $\mathbb{P}(|\widehat{T}_N - \ES_{\alpha}(X)| \geq 1)$ (plug-in estimator) and $\mathbb{P}(|\widehat{S}_N - \ES_{\alpha}(X)| \geq 1)$ (the proposed estimator uses $\beta_1 = 0.5, \beta_2 = 0.6$ and $m=250$) for varying values of $N$, showcasing the exponential rate for $\widehat{S}_N$ and the lack thereof for $\widehat{T}_N$. Here $\alpha=0.1$, $X$ is Pareto distributed with $\lambda=2.2$, and the deviation probabilities are estimated using $10^6$ many simulations. The middle image showcases the histogram of estimated values across simulations for $N=3250$, with the blue vertical lines depicting the values $\ES_{\alpha}(X)-1, \ES_{\alpha}(X), \ES_{\alpha}(X)+1$. The right hand image shows the tail behavior of the estimators. Among the $10^6$ simulations for $N = 3250$, there are 13637 errors larger than one for $\widehat{T}_N$ (largest realized value 98.98), and 1568 for $\widehat{S}_N$ (largest realized value 7.83).}\label{fig:comparisonestimator}
\end{figure}

A simple illustration of the performance of our estimator $\widehat{S}_N$ compared with the  plug-in estimator $\widehat{T}_N$ is provided in Figure \ref{fig:comparisonestimator}. 
Beyond the rate of estimation, another interesting aspect which can be seen from the illustrated example is that the proposed estimator never fails drastically.
Indeed, while the plug-in estimator $\widehat{T}_N$ sometimes yields extreme values (e.g.\ a larger than 1000\% deviation from the true value even with $N=3250$), the proposed estimator $\widehat{S}_N$ is much more robust in this regard, never deviating more than 45\% from the true value across all $10^6$ many simulations. 

Figure \ref{fig:comparisonestimator} further underlines the intuition given in Remark \ref{rem:bias}: whenever the values of the plug-in estimator $\widehat{T}_N$ lie within the typical regions, then $\widehat{S}_N\approx \widehat{T}_N$ (middle image).
Conversely, the right image showcases that in the presence of outliers in the data, $\widehat{S}_N$ outperforms $\widehat{T}_N$.

\vspace{0.5em}
For completeness, let us note that in the regime  when $\varepsilon$ is small, the `second order' term $\mathcal{O}(L\varepsilon^2)$ may be neglected and the estimate in Theorem \ref{thm:ES} is of the same order as the (in general too) optimistic guess from the central limit theorem in \eqref{eq:exp.rate}.
To state the result formally, let $c_0,\dots,c_3$ be the constants appearing in Theorem \ref{thm:ES}.

\begin{corollary}\label{cor:regimeswitch}
	In the setting  of Theorem \ref{thm:ES}, if in addition  $\varepsilon\leq \frac{ \sigma_{\ES_\alpha(X)}}{c_2L}$, then
	\begin{align}
	\label{eq:intro.small.eps}
	 \P\left( \left| \widehat{S}_N-\mathrm{ES}_\alpha(X) \right| \geq 2  \varepsilon\sigma_{\ES_\alpha(X)}    \right)
	 \leq \exp\left(-c_3 N\varepsilon^2\right).
	\end{align}
\end{corollary}

Combining the conditions on $\varepsilon$ from Theorem \ref{thm:ES} and Corollary \ref{cor:regimeswitch} implies that  $\varepsilon\lesssim \min\{ \sqrt{\alpha},\frac{ \sigma_{\ES_\alpha(X)}}{L} \}$.
While the value of $\frac{ \sigma_{\ES_\alpha(X)}}{L} $ heavily depends on the underlying distribution of $X$ and can be significantly smaller than $\sqrt{\alpha}$ in general, we shall show in Section \ref{sec:Pareto} that if $X$ has a Pareto distribution, then $\frac{ \sigma_{\ES_\alpha(X)}}{L}\lesssim \sqrt{\alpha}$.
Thus, for the Pareto distribution, the additional restriction on $\varepsilon$ in Corollary \ref{cor:regimeswitch} is (essentially) satisfied automatically.

In the reminder of this paragraph, we will consider random variables $X$ that have a non-negative density $f_X$.
In that case,  
\[D(\alpha):= \frac{d}{du} \VAR_u(X) |_{u=1-\alpha}
= \frac{1}{f_X(F_X^{-1}(1-\alpha))} \]
is the local Lipschitz constant of $u\mapsto \VAR_u(X)$ at $u=1-\alpha$.
In Table \ref{table:parameters}, we numerically compute $D(\alpha)$ and $\sigma_{\ES_\alpha}$ for various choices of $\alpha$ and distributions. 

\begin{table}
	\label{tableComparison}
	\centering
	\begin{tabular}{@{}l c l c c c c c c c c c}\toprule
        \multirow{2}{*}{Distribution} & & \multirow{2}{*}{Parameters} & & \multicolumn{2}{c}{$\alpha=0.1$} & & \multicolumn{2}{c}{$\alpha=0.05$}  & &  \multicolumn{2}{c}{$\alpha=0.01$} \\
        & & & & $D(\alpha)$ & $\sigma_{\ES_{\alpha}}$ & & $D(\alpha)$ & $\sigma_{\ES_{\alpha}}$ & & $D(\alpha)$ & $\sigma_{\ES_{\alpha}}$ \\ \midrule
        Normal & & $\mu=0, \sigma=1$ & & 5.70 & 1.93 & & 9.70 & 2.47 & & 37.5 & 4.59 \\
        Student-t & & $\nu=5$ & & 7.79 & 3.88 & & 15.7 & 5.99 & & 91.7 & 17.2\\
        Logistic & & $\alpha=0, \beta=1$ & & 11.1 & 4.53 & & 21.1 & 6.36 & & 101 & 14.2 \\
        Lognormal & & $\mu=0, \sigma=1$ & & 20.5 & 14.9 & & 50.2 & 25.2 & & 384 & 82.0 \\
        Pareto & & $x_m=1, \lambda=2$ & & 15.8 & $\infty$ & & 44.7 & $\infty$ & & 500 & $\infty$ \\
        Pareto & & $x_m=1, \lambda=4$ & & 4.45 & 3.18 & & 10.6 & 5.39 & & 79.1 & 18.2 \\
        Exponential & & $\lambda=1$ & & 10 & 4.35 & & 20 & 6.24 & & 100 & 14.1 \\\bottomrule
	\end{tabular}
	\caption{The local Lipschitz constant and the standard deviation $\sigma_{\ES_\alpha}$ for different values of $\alpha$ and different distributions. The values are reported up to three significant digits and, for $\sigma_{\ES_\alpha}$, based on numerical integration.}
    \label{table:parameters}
\end{table}

Note that the Lipschitz constant $L = L(\alpha)$ appearing in Theorem~\ref{thm:ES} is defined as 
\[
L = \max_{\beta \in [\alpha/2,\, 2\alpha]} D(\beta).
\]
The interval \([\alpha/2, 2\alpha]\) was chosen for analytical convenience in order to obtain clean bounds. 
In practice, however, the value of \(L\) is often much closer to \(D(\alpha)\). 
For this reason, we report \(D(\alpha)\) instead of \(L(\alpha)\) in Table \ref{table:parameters}.

\subsection{Adversarial robustness}

In addition to its favourable statistical behavior, our estimator $\widehat{S}_N$ possesses another advantage: it is \emph{adversarially robust}. 
The latter means that even if an adversary can modify a certain fraction of the data $X_1,\dots, X_N$ arbitrarily (even under knowledge of the real data and knowledge of the estimation procedure), one still has guarantees that the estimator $\widehat{S}_N$ performs well. 
It is worth noting that this is in stark contrast to the behavior of the classical plug-in estimator, where (regardless of the distribution of $X$) the modification of even a \emph{single} sample $X_i$ can completely distort the resulting estimate.

\begin{theorem}\label{thm:adversarial}
	Let the assumptions in Theorem \ref{thm:ES} hold.
	Then there are absolute constants $C_0,C_1,C_2$ such that even if at most $C_0 N\varepsilon^2$ of the $N$ samples are maliciously modified,  
	\[ \P\left( \left| \widehat{S}_N-\mathrm{ES}_\alpha(X) \right| \geq   \varepsilon\sigma_{\ES_\alpha(X)}  + C_1 L \varepsilon^2   \right)
	\leq \exp\left(-C_2 N\varepsilon^2\right).\]
\end{theorem}

We show that  $C_0:=\frac{1}{140}$, $C_1=10$, and $C_2:=\frac{1}{2800}$ are valid choices, but believe that the optimal values of these constants are considerably better. 

Finally, let us note that if a sufficiently large number of data points are corrupted (say, all of them), then clearly there  can be no procedure that estimates the expected shortfall in a satisfactory manner. 
In our context, denoting by $K$ the number of corrupted data points, we require that $K \leq C_0N\varepsilon^2$; or, in other words, if we expect that $K$ data points are corrupted, the lower bound on the level of accuracy of the estimator is $\varepsilon \geq \sqrt{ \frac{K}{C_0N}}$.
As it happens, this dependence of the error on the number of corrupted data turns out to be  minimax-optimal in the class of distributions satisfying the assumption of Theorem \ref{thm:adversarial}:
We show in Proposition \ref{prop:adversiarl} that  no estimator can have  a better error rate than $ \mathcal{O}(\sigma_{\ES_\alpha(X)} \sqrt{K/N})$, where $K$ is the number of corrupted data points.

\subsection{Related literature}\label{subsec:relatedliterature}
The estimation of Expected Shortfall is a classical and widely studied topic, motivated both by regulatory requirements in financial practice and by long-standing theoretical interest. 
A substantial literature has investigated non-parametric and empirical approaches to ES estimation in the i.i.d.\ setting (see, e.g., \cite{acerbi2007coherent,brazauskas2008estimating,cont2010robustness,hill2015expected,jones2003empirical,yamai2002comparative}). 
In parallel, numerous works address more realistic data-generating mechanisms—such as stochastic volatility, time-series dependence, and heavy-tail behavior—including dynamic semiparametric forecasting models (see, e.g., \cite{patton2019dynamic}), extreme-value based estimation (see, e.g., \cite{drees2004maximum,kumar2017}), and non-parametric methods for dependent data (see, e.g., \cite{moutanabbir2024}). 
Further important developments concern the assessment and evaluation of Expected Shortfall estimates through backtesting (see, e.g., \cite{acerbi2014back,berens2018,FisslerZiegel2016,moldenhauer2017backtesting}), as well as quantification of estimation error and finite-sample uncertainty (see, e.g., \cite{aichele2025coherent,bartl2023nonasymptotic,lonnbark2013role,pflug2010asymptotic,pitera2018}).
We further refer to \cite{kamronnaher2024estimating,nadarajah2014estimation} for surveys and to \cite{mcneil2015quantitative,wuthrich2023statistical} for comprehensive monographs.

Moving on to literature more closely related to this paper, note that
for the particular choice of $\beta_1=\beta_2=1/2$, the estimator $\widehat{S}_N$ is the median of the block estimators $(\widehat{T}_{I_j})_{j=1}^n$, which we call the median of blocks estimator for short.
In the context of estimating the mean of a random variable, this median of blocks estimator constitutes the analogue to the so-called median-of-means estimator, whose optimal statistical performance is well understood (see, e.g.\ \cite{alon1996space,nemirovskij1983problem,jerrum1986random}). Contrary to mean estimation, properties for estimating the expected shortfall have only been obtained more recently.
For instance, proving a central limit theorem for the empirical plug-in estimator $\widehat{T}_N$ was obtained in \cite{brazauskas2008estimating,pflug2010asymptotic}, and, to the best of the authors' knowledge, no statistically optimal finite sample guarantees for any expected shortfall estimator are available in the literature.

Another topic this paper relates to is that of stability and (adversarial) robustness of estimators.  Adversarial robustness of data-driven systems (see, e.g., \cite{carlini2019evaluating, goodfellow2018making,art2018}) can be regarded as a strong form of robustness against outliers (see, e.g., \cite{kim2012robust,rousseeuw2011robust, rousseeuw2005robust}), and has recently gathered increasing attention, particularly in fields adjacent to machine learning. In line with this, Theorem \ref{thm:adversarial} provides statistical guarantees even when a small fraction of the data may be erroneous or even manipulated. We believe such guarantees are particularly helpful if growing amounts of data are used which make strict quality control harder to enforce.
Finally, we emphasize that adversarial robustness can also be seen as a stability property of estimators with respect to perturbations in the data (see, e.g.,~\cite{andrews1986stability,christmann2018total,eckstein2023dimensionality}), particularly since modifying a part of the data corresponds to a perturbation of the data distribution with respect to the total variation norm. In this respect, Theorem \ref{thm:adversarial} may also be interpreted not just as robustness with respect to errors in the data, but also as robustness with respect to rare structural breaks in the data-generating distribution.


\subsection{Extensions}
\label{sec:extensions}

While the focus of this paper is the construction of a statistically optimal nonparametric estimator, we want to emphasize that there are several further aspects (beyond statistical) which need to be taken into account towards a practically useful estimation procedure. Indeed, when estimating expected shortfall risks in practice, $\alpha$ may be very low, for instance $\alpha=0.01$, and one has relatively little non-i.i.d.~data, perhaps $N=250$, which means there are around three data points in the tail. Clearly, in such situations, purely nonparametric estimators must fail. Hence it is a relevant question whether the ideas developed in this paper can lead to new insights in these practical regimes as well. While a full treatment of this question would go far beyond the scope of this paper, we shortly discuss natural directions in the following.

First, going towards non-i.i.d.~data, a natural next step towards more realistic assumptions can be pursued in the form of mixing conditions (see, e.g., \cite{chen2008nonparametric,scaillet2004nonparametric}). We believe that applying the proposed estimator (or slight variations thereof) is still very valuable in such situations and we outline a first approach in Section \ref{sec:non.iid}.

Further, in relation to practical tools employed in finance, it is worth mentioning that the proposed estimator in Definition \ref{def:proposed} and the median of blocks estimator can be understood in terms of resampling methods, particularly non-parametric bootstrapping (cf.~\cite{horowitz2019bootstrap}). Indeed, the splitting method for the intervals in Definition \ref{def:proposed} can be regarded as resampling $n$ many samples of size $m$ (not independently as usually done in bootstrapping, but non-overlapping). This resampling perspective is still in line with recent improvements proposed to the median-of-mean estimator using overlapping intervals (cf.~\cite{minsker2022u, minsker2023efficient}). In this direction, practical improvements to the estimator proposed in Definition \ref{def:proposed} may be possible simply be taking larger non-overlapping intervals (whether theoretically justified or not), which significantly decreases the minimal number of data points necessary to apply the estimator at all.

In a similar spirit, the estimator proposed in Definition \ref{def:proposed} could be combined with alternatives to the empirical plug-in estimator, such as parametric methods (cf.~ \cite[Section~3]{nadarajah2014estimation}) or structural assumptions on the tail behavior of the distribution. This means, one could simply use other (parametric) estimators compared to $\widehat{T}_N$ and/or $\widehat{T}_{I_j}$ within the definition (with potentially different choices of $m$ and $n$).
Such approaches become particularly relevant in situations where only a very small number of observations fall into the tail region relevant for computing $\mathrm{ES}_\alpha(X)$ (e.g., $N=250$ and $\alpha=0.01$).
In these cases, purely non-parametric estimation becomes unreliable, and incorporating parametric information is often necessary to achieve reasonable accuracy.
Developing a principled combination of the methods for improving statistical guarantees proposed in this article with parametric techniques for extreme-tail estimation appears to be a highly relevant direction for future research.

Another important point is that while we numerically test our estimator in Section \ref{sec:numerics} in terms of the statistical learning task described in this paper, estimators in practical settings with non-i.i.d.~(and sometimes even non-stationary) data are typically evaluated using backtesting, which itself may require certain model assumptions in the case of the expected shortfall risk (cf.~\cite[Section 5]{embrechts2014academic}).

Finally, we believe that our methods extend to other risk measures beyond the expected shortfall, in particular to the class of distortion risk measures (see, e.g., \cite{wirch2001distortion}) since the representation of the expected shortfall as a distortion risk measures is a crucial property used for the main results in this paper.

\subsection{Structure of the paper}
The remainder of the paper is structured as follows. 
The following Section \ref{sec:mainproofs} contains the proofs of Theorems \ref{thm:ES} and \ref{thm:adversarial}. 
In Section \ref{sec:onoptimality} we focus on the results regarding optimality of the estimates.
The focus of Section \ref{sec:Pareto} is on Pareto distributions, in which case an explicit lower bound for the plug-in estimator is also established.
In Section \ref{sec:numerics} we present several numerical experiments.

\section{Proofs of Theorem \ref{thm:ES} and Theorem \ref{thm:adversarial}}\label{sec:mainproofs}

For shorthand notation, we write $\sigma^2_{\ES_\alpha}$ instead of $\sigma^2_{\ES_\alpha(X)}$ and $F$ instead of $F_X$.
Moreover, recall that $m=\lceil  \frac{11}{\varepsilon^2} \rceil$, that $\varepsilon\leq \frac{1}{6}\sqrt{\alpha}$  where   $\alpha < \frac{1}{2}$, and that  $n=\lfloor \frac{N}{m}\rfloor$.
In particular, the assumptions that $\varepsilon\leq 1$  and that $N\varepsilon^2\geq 140$ immediately imply that
\begin{align}
\label{eq:estimates.m.n} 
 \frac{12}{\varepsilon^2}\geq m \geq 11,
\quad\text{and}\quad
n \geq \frac{ N\varepsilon^2}{14} \geq 10. 
\end{align}
We shall use these facts many times.

To make the analysis more transparent, we will assume (in this section only) that $F$ is continuous and strictly increasing.
This can be done without loss of generality: Otherwise, instead of $X$ consider $Y^\eta:=X+\eta G$ where $G$ is a standard Gaussian random variable and $\eta>0$ is arbitrarily small. 
Clearly
\[\ES_\alpha(Y^\eta)\to \ES_\alpha(X)\quad\text{and}\quad
\sigma^2_{\ES_\alpha(Y^\eta)}\to \sigma^2_{\ES_\alpha(X)}\]
 as $\eta\to 0$.
In a similar manner, for every arbitrarily small $\gamma>0$, if $\eta$ is sufficiently small, then the samples from $X$ and from $Y^\eta$ are almost the same, i.e., $|Y_i^\eta-X_i|\leq \gamma$ for all $i$ with arbitrarily high probability.
Thus computing the estimator $\widehat{S}_N$ using the samples from $X$ or from $Y$ will results in an error of at most $\gamma$, which is arbitrarily small. 

The proof of Theorem \ref{thm:ES} is based on the following standard fact about Binomial concentration.
For the sake of completeness, we provide the short proof.

\begin{lemma}
\label{lem:binomial.concentration}
	Let $y_0\in\mathbb{R}$, let  $\delta\geq0$, and set  $(Y_j)_{j=1}^n$  to be independent random variables satisfying $	\P( | Y_j- y_0| > \delta ) \leq \frac{1}{10}$ 	for every $j=1,\dots,n$.
	Then,  
	\[\P\left( \left| \{ j\in\{1,\dots,n\} : |  Y_j - y_0 |>\delta \right| \} \geq 0.3 n \right) 
	\leq \exp\left( - \frac{n}{50} \right).   \]
\end{lemma}
\begin{proof}
	An application of Hoeffding's inequality (see, e.g., \cite[Theorem 2.8]{boucheron2003concentration}) implies that for every $\lambda>0$, with probability at least $ 1-\exp( -\frac{\lambda^2}{2n})$, 
	\[
	\sum_{j=1}^n \eins_{\{|Y_j-y_0|>\delta\}}
	<  \sum_{j=1}^n  \E[\eins_{\{|Y_j-y_0|>\delta\}}] + \lambda.
\]
	To complete the proof, set  $\lambda=\frac{2}{10}n$.
\end{proof}

	Let us recall that 
	\[\widehat{T}_{I_j} 
	= \ES_\alpha\left( \widehat{F}_{I_j} \right)
	= \frac{1}{\alpha} \int_{1-\alpha}^1 \VAR_u \left( \widehat{F}_{I_j} \right) \,du\]
	is the plug-in estimator of the sample restricted to the block $I_j$.
	Moreover, as the blocks $(I_j)_{j=1}^n$ are disjoint, the family $(\widehat{T}_{I_j})_{j=1}^n$ is independent.
	Hence, with Lemma \ref{lem:binomial.concentration} in mind, if we can show that 
	\begin{align}
	\label{eq:need.to.show}
	\P\left(	 \left| \widehat{T}_{I_j} -  \ES_\alpha(X) \right| >  \varepsilon \sigma_{\ES} +  10 L \varepsilon^2\right) 
 \leq \frac{1}{10},
	\end{align}
	then letting $0.3n\leq k \leq \ell\leq 0.7n$, Lemma \ref{lem:binomial.concentration} implies that with exponentially high probability, 
	\begin{equation}\label{eq:sandwich} \widehat{T}_{(k)} , \widehat{T}_{(\ell)} \in \left [ \ES_\alpha(X) - \varepsilon \sigma_{\ES} -  10 L \varepsilon^2 \, , \, \ES_\alpha(X) + \varepsilon \sigma_{\ES} +  10 L \varepsilon^2 \right],\end{equation}
	where we recall that $(\widehat{T}_{(j)})_{j=1}^n$ are the order statistics of $(\widehat{T}_{I_j})_{j=1}^n$.
	In particular, since $\widehat{S}_N$ is sandwiched between $\widehat{T}_{(k)}$ and $\widehat{T}_{(\ell)}$,  this will imply the statement in Theorem \ref{thm:ES}.

	\vspace{0.5em}
	
	In order to prove \eqref{eq:need.to.show},  we make use of the following representation of the expected shortfall as a so-called distortion risk measure (see, e.g., \cite{sereda2010distortion,tsukahara2014estimation}):
	Setting
	\[ \psi\colon[0,1]\to[0,1],\qquad x\mapsto \min\left\{ \frac{1-x}{\alpha} , 1\right\},\]
	then, for every random variable $Y$,
	\begin{align}
	\label{eq:es.distrition}
	\ES_\alpha(Y)
	=\int_{-\infty}^0 \psi(F_Y(t)) - 1\, d t + \int_0^\infty \psi(F_Y(t))\,dt.
	\end{align}
	As our arguments for showing that \eqref{eq:need.to.show} holds crucially rely on  \eqref{eq:es.distrition}, we give the short proof of \eqref{eq:es.distrition} for the convenience of the reader.
	
	\begin{proof}[Proof of \eqref{eq:es.distrition}]
	In a first step, assume that  $F_Y(0) \leq 1-\alpha$, i.e.\ $\P(Y>0)\geq \alpha$.
	Then $\psi(F_Y(t))=1$ for every $t< 0$ and the first integral in \eqref{eq:es.distrition} is equal to zero.
	An application of Fubini's theorem shows that
	\begin{align*}
	\ES_\alpha(Y)
	= \frac{1}{\alpha}\int_{1-\alpha}^1 F_Y^{-1}(u)\,du
	&=\frac{1}{\alpha}\int_{1-\alpha}^1 \int_0^\infty \eins_{[0, F_Y^{-1}(u)]}(t)\,dt\,du \\
	&=\int_0^\infty \frac{1}{\alpha}\int_{1-\alpha}^1  \eins_{(F_Y(t),1]}(u)\,du\,dt\\
	&=\int_0^\infty \frac{1}{\alpha}\min\left\{ 1-F_Y(t), \alpha\right\}\,dt
	=\int_0^\infty  \psi(F_Y(t)) \,dt,
	\end{align*}
	which proves \eqref{eq:es.distrition} in case that $\P(Y>0) \geq \alpha$.
	
	As for the general case, let $y$ be a constant that satisfies $\P(Y+y>0)\geq \alpha$ and set $Y':=Y+y$.
	The proof then follows from a short computation using the fact that  \eqref{eq:es.distrition} holds for $Y'$ (by the first step), that  $\ES_\alpha(Y')=\ES_\alpha(Y)+y$, and that $F_{Y'}(t)=F_Y(t-y)$.
	\end{proof}	
	
	To control the difference between $\widehat{T}_{I_j}$ and $\ES_\alpha(X)$, note that $\psi$ is almost everywhere differentiable with $\psi'=-\frac{1}{\alpha}\eins_{[1-\alpha,1]}$.
	Therefore, a first order Taylor expansion yields a decomposition  of the estimation error
	\begin{align*}
	\widehat{T}_{I_j}- \ES_\alpha(X) 
	=\ES_\alpha \left( \widehat{F}_{I_j} \right)- \ES_\alpha(X) 
	&=\int_\mathbb{R} \psi\left( \widehat{F}_{I_j}(t) \right) -\psi(F(t))\,dt\\
	&= \mathcal{L}_j + \mathcal{E}_j
	\end{align*}
	into a \emph{linearised term}  and an \emph{error term}
	\begin{align*}
	\mathcal{L}_j
	&:=\int_\mathbb{R} \psi'(F(t)) \left(\widehat{F}_{I_j}(t)-F(t) \right) \,dt, \\
	\mathcal{E}_j
	&:= \int_\mathbb{R} \Big( \psi'\big(F(t)+ \xi (\widehat{F}_{I_j}(t)-F(t))\big)-\psi'(F(t)) \Big) \left( \widehat{F}_{I_j}(t)-F(t) \right) \,dt,
	\end{align*}
	where  $\xi\in (0,1)$ is some midpoint (that depends on $t$, $\widehat{F}_{I_j}$ and $F$). 
	
	In the following we analyse the linearised term and the error term  separately.
		
\begin{lemma}
\label{lem:linear}
	We have that 
	\[\P \left(|\mathcal{L}_j|> \varepsilon \sigma_{\ES_\alpha} \right)\leq \frac{1}{11}.\]
\end{lemma}
\begin{proof}
	We first claim that $\E[ \mathcal{L}_j^2 ]= \sigma^2_{\ES} / m$.
	To that end, note that
	\[\mathcal{L}_j^2 
	=\int_{\mathbb{R}}\int_{\mathbb{R}} \psi'(F(t))\psi'(F(s)) \left( \widehat{F}_{I_j}(t)-F(t) \right) \left( \widehat{F}_{I_j}(s)-F(s) \right)\, dt ds.\]
	Moreover, since  $\E[\widehat{F}_{I_j}(t)]=F(t)$ for every $t\in\mathbb{R}$, an application of  Fubini's theorem shows that 
	\begin{align*}
	\E[ \mathcal{L}_j^2 ]
	&=\int_{\mathbb{R}}\int_{\mathbb{R}}  \psi'(F(t))\psi'(F(s)) \C\left[\widehat{F}_{I_j}(t), \widehat{F}_{I_j}(s) \right] \, dt ds \\
	&= \frac{1}{\alpha^2} \int_{F^{-1}(1-\alpha)}^\infty \int_{F^{-1}(1-\alpha)}^\infty  \C\left[\widehat{F}_{I_j}(t), \widehat{F}_{I_j}(s) \right] \, dt ds,
	\end{align*}
	where the second equality holds because $\psi'=-\frac{1}{\alpha} \eins_{[1-\alpha,1]}$ and $F(t)<1-\alpha$ for every $t<F^{-1}(1-\alpha)$. 
	Next, by independence of the sample $(X_i)_{i\in I_j}$, 
	\begin{align*}
	\C\left[\widehat{F}_{I_j}(t), \widehat{F}_{I_j}(s)\right]
	&=\frac{1}{m^2} \sum_{i\in I_j} \E\left[ (\eins_{(-\infty,t]}(X_i)  -F(t) ) (\eins_{(-\infty,s]}(X_i) - F(s))\right]\\
	&=\frac{ F(\min\{t, s\}) -F(t)F(s) }{m}.
	\end{align*}
	Therefore, it follows from the definition of $\sigma^2_{\ES_\alpha}$ that  	$\E[ \mathcal{L}_j^2 ] =\sigma^2_{\ES_\alpha} / m$.
	
	To complete the proof, we use  Chebychev's inequality which guarantees that 
	\[ \P\left(|\mathcal{L}_j|> \varepsilon \sigma_{\ES} \right)
	\leq \frac{\E[\mathcal{L}_j^2]}{(\varepsilon \sigma_{\ES})^2}
	= \frac{\sigma^2_{\ES}}{m  \varepsilon^2 \sigma_{\ES}^2}
	\leq \frac{1}{11},\]
	where the last inequality holds since  $m \geq \frac{11}{\varepsilon^2}$ (see \eqref{eq:estimates.m.n}).
\end{proof}

\begin{lemma}
\label{lem:error}
	Set $\gamma:= \frac{5}{4} \varepsilon\sqrt\alpha$.
	Then we have that
\[ \P\left(|\mathcal{E}_j|>  \frac{ 3\gamma  }{\alpha} \left( F^{-1}(1-\alpha+\gamma) -F^{-1}(1-\alpha- \gamma) \right)\right)
	\leq \frac{1}{110}.\]
\end{lemma}
\begin{proof}	
	Note that $\gamma\leq \frac{5}{24} \alpha$ by the assumption that $\varepsilon< \frac{1}{6}\sqrt{\alpha}$, and let $t_0,t_1\in\mathbb{R}$ satisfy that
	\[   F(t_0)= 1-\alpha - \gamma ,\quad F(t_1)= 1-\alpha+ \gamma . \]
	Such $t_0,t_1$ exist since $F$ is strictly increasing and continuous in this section.

	\vspace{0.5em}
	\noindent
	{\em Step 1:}
	We claim that 
	\[ \mathcal{A}
	:= \left\{ \left|  \widehat{F}_{I_j} (s) - F(s) \right| \leq \gamma  \text{ for } s = t_0,t_1 \right\} \]
	satisfies $\P( \mathcal{A}) \leq  \frac{1}{110}$.
	To that end, recall that  Bernstein's inequality 	(see, e.g., \cite[Theorem 6.12]{steinwart2008support}) guarantees that if $Y_1,\dots, Y_m$ are i.i.d.~random variables which take values between 0 and 1, then for every $\lambda\geq 0$, with probability at least $1-2\exp(-\lambda)$,
	\[  \left| \frac{1}{m}\sum_{i=1}^m Y_i - \E[Y_1]  \right|  
	\leq \sqrt \frac{2 \V[Y_1] \lambda }{m} + \frac{2\lambda}{3m} .\]
	Applied to $Y_i= 1_{(-\infty,s]}(X_i)$ (for $s=t_0,t_1$) and $\lambda:=\ln(440)$, it follows that with probability at least $1-\frac{1}{220}$, 
	\[  \left| \widehat{F}_{I_j} ( s) - F(s)  \right|  
	\leq  \sqrt \frac{2F(s)(1-F(s)) \ln(440) }{m} + \frac{2\ln(440)}{3m}
	=:\varphi(s) .\]
	Thus, it remains to show that $\varphi(s)\leq\gamma$ for $s=t_0,t_1$.
	
	To that end, fix $s\in\{t_0,t_1\}$ and note that since   $\gamma\leq \frac{5}{24}\alpha$,
	\[F(s)(1-F(s)) \leq \alpha+\gamma
	\leq  \frac{29}{24}\alpha.\]
	Therefore, using that $m\geq \frac{11}{\varepsilon^2}$ and that $\varepsilon^2\leq \frac{1}{6}\varepsilon\sqrt\alpha$,
	\begin{align*}
	\varphi(s)
	&\leq  \sqrt{ \frac{2 \cdot 29 \cdot  \alpha  \varepsilon^2  \cdot \ln(440) }{24\cdot 11} } +  \frac{2\cdot \ln(440)  \varepsilon^2}{33} \\
	&\leq \varepsilon\sqrt{\alpha} \cdot \left( \sqrt{\frac{ 2\cdot 29 \cdot \ln(440) }{24\cdot 11} } + \frac{2\cdot \ln(440) }{33 \cdot 6} \right)
	\leq\varepsilon\sqrt{\alpha}  \cdot \frac{5}{4} = \gamma .
	\end{align*}

	\vspace{0.5em}
	\noindent
	{\em Step 2:}
	We claim that in the event $\mathcal{A}$, 
	\begin{align}
	\label{eq:2}
	|\mathcal{E}_j|
	\leq \frac{1}{\alpha} \int_{t_0}^{t_1 } \left| F(t)-\widehat{F}_{I_j}(t) \right|\,dt.
	\end{align}
	Indeed, recall that 
	\[\mathcal{E}_j
	= \int_\mathbb{R} \Big( \psi'\big(F(t)+ \xi (\widehat{F}_{I_j}(t)-F(t))\big)-\psi'(F(t)) \Big) \left( \widehat{F}_{I_j}(t)-F(t) \right) \,dt \]
	and observe that for $t<t_0$ or $t>t_1$, the term in the first bracket is equal to zero.
	Indeed, consider $t<t_0$ and note that by the strict monotonicity of $F$ and the definition of $t_0$, $F(t) < 1- \alpha-\gamma$.
	Moreover, by the  monotonicity of $\widehat{F}_{I_j}$,
	\[\widehat{F}_{I_j}(t)
	\leq \widehat{F}_{I_j}(t_0)
	\leq F(t_0)+\gamma
	\leq 1-\alpha,\]
	where the second inequality follows from the definition of $\mathcal{A}$. 
	Since $\psi'=-\frac{1}{\alpha} \eins_{[1-\alpha,1]}$, it therefore follows that for  $\xi\in(0,1)$, 
	\[ \psi'\big(F(t)+ \xi (\widehat{F}_{I_j}(t)-F(t))\big)-\psi'(F(t)) =0-0=0. \]

	In a similar manner, for every $t>t_1$, $F(t)>1-\alpha+\gamma$ and $\widehat{F}_{I_j}(t)\geq 1-\alpha$,
	and thus for $\xi \in (0, 1)$,
	\[ \psi'\big(F(t)+ \xi (\widehat{F}_{I_j}(t)-F(t))\big)-\psi'(F(t)) =1-1=0 \]
	and \eqref{eq:2} follows. 
	
		\vspace{0.5em}
	\noindent
	{\em Step 3:}
	First observe that  in the event $\mathcal{A}$, for every $t\in[t_0,t_1]$,
	\begin{align}
	\label{eq:diff.F.Fm}
	\left|F(t)-\widehat{F}_{I_j}(t) \right|
	\leq 3 \gamma.
	\end{align}
	Indeed, fix $t\in[t_0,t_1]$.
	Then, by the monotonicity of $\widehat{F}_{I_j}$ and the definition of $\mathcal{A}$, 
	\[  \widehat{F}_{I_j}(t)
	\leq \widehat{F}_{I_j}(t_1)
	\leq F(t_1)+\gamma
	\leq F(t)+2\gamma + \gamma .\]
	The same argument implies that $\widehat{F}_{I_j}(t)
	\geq F(t)-3\gamma$ which shows \eqref{eq:diff.F.Fm}.

	Finally, using  \eqref{eq:2} and \eqref{eq:diff.F.Fm}, we conclude that in the event $\mathcal{A}$,
	\begin{align*}
	|\mathcal{E}_j|
	\leq \frac{1}{\alpha} \int_{t_0}^{t_1 }3  \gamma \,dt
	&= \frac{3(t_1-t_0) \gamma}{\alpha},
	\end{align*}
	and the proof is completed by the definition of $t_0$ and $t_1$.
\end{proof}	

\begin{proof}[Proof of Theorem \ref{thm:ES}]
	Recall that 
	\[ \widehat{T}_{I_j} - \ES_\alpha(X) = \mathcal{L}_j + \mathcal{E}_j\]
	and that by Lemma \ref{lem:linear},  $\P(|\mathcal{L}_j|\geq \varepsilon\sigma_{\ES_\alpha})\leq \frac{1}{11}$.
	To control the term $\mathcal{E}_j$, set $\gamma:=\frac{5}{4}\varepsilon\sqrt{\alpha}$ so that  Lemma \ref{lem:error} guarantees that with probability at least $1-\frac{1}{110}$,
	\[ |\mathcal{E}_j| \leq   \frac{ 3\gamma  }{\alpha} \left( F^{-1}(1-\alpha+\gamma) -F^{-1}(1-\alpha- \gamma)	 \right)
	=:(\ast) .\]
	Since $\gamma\leq \frac{1}{2}\alpha$ (by the assumption that $\varepsilon\leq \frac{1}{6}\sqrt\alpha$) and  $F^{-1}$ is $L$-Lipschitz on  $[1-2\alpha,1-\frac{1}{2}\alpha]$, it follows that
	\[ (\ast)  
	\leq\frac{ 6 \gamma^2 L }{\alpha}
	\leq 10 \varepsilon^2L .\] 
	
	Combining the estimates on $\mathcal{L}_j$ and $\mathcal{E}_j$ shows that
	\[\P\left(\left| \widehat{T}_{I_j} - \ES_\alpha(X) \right| \geq  \varepsilon\sigma_{\ES_\alpha} +  10 L \varepsilon^2  \right)
	\leq\frac{1}{10}.\]
 	Therefore, \eqref{eq:need.to.show} holds, and the proof is completed by an application of Lemma \ref{lem:binomial.concentration}.
\end{proof}

\begin{proof}[Proof of Theorem \ref{thm:adversarial}]
	The proof requires only minor modification to the proof of Theorem \ref{thm:ES}.
	For shorthand notation, set $\delta:= \varepsilon\sigma_{\ES_\alpha(X)}  + 10 L \varepsilon^2$.
	
	In a first step, suppose that the data $X_1,\dots,X_N$ are not modified, i.e.\ we work in the setting of Theorem \ref{thm:ES}.
	Then, by Lemma \ref{lem:linear} and Lemma \ref{lem:error}, $\P(|\widehat{T}_{I_j}- \ES_\alpha(X)|>\delta)\leq\frac{1}{10}$ for all $j$.
	Hence, the same arguments as presented in the proof of  Lemma \ref{lem:binomial.concentration} together with the estimates on $n$ and $m$ (see \eqref{eq:estimates.m.n}) imply that with probability at least $1-\exp(-\frac{1}{14\cdot 200}N\varepsilon^2)$, 
	\begin{align}
	\label{eq:adversiral}
	| \{j : |\widehat{T}_{I_j}- \ES_\alpha(X) | >\delta\}| 
	< \sum_{j=1}^n \P(|\widehat{T}_{I_j}- \ES_\alpha(X)|>\delta) + \frac{n}{10}
	\leq \frac{2n}{10}.
	\end{align}
	
	In a next step, recall that at most $K\leq \frac{1}{140} N \varepsilon^2$ samples are modified and that $\frac{1}{140}N\varepsilon^2\leq \frac{n}{10}$, see \eqref{eq:estimates.m.n}.
	Therefore, if at most $K$ of the samples $X_1,\dots,X_N$ are modified, this can yield at most a modification of $\widehat{T}_{I_j}$ on $\frac{n}{10} $ of the blocks $I_1,\dots, I_n$.
	Hence, in the high probability event in which \eqref{eq:adversiral} holds, we still have that for more than $\frac{7}{10}n$ of the $j$'s, $|\widehat{T}_{I_j}- \ES_\alpha(X) |\leq \delta$ and thus \eqref{eq:sandwich} holds, which yields the claim.
\end{proof}

For completeness, let us state the following immediate consequence of  Theorem \ref{thm:adversarial} on the consistency of $\widehat{S}_N$ (as $N\to\infty$) in the presence of corrupted data.

\begin{corollary}
\label{cor:adversarial}
	Suppose that $\sigma_{\ES_\alpha}^2$ is finite, that $u\mapsto \VAR_u(X)$ is Lipschitz continuous on a neighbourhood of $1-\alpha$,  that at most $K_N<N$ data points are maliciously modified, and that $\lim_{N\to\infty} K_N/N=0$.
	Then, setting $\varepsilon_N:=12\sqrt{K_N/N}$, it follows that
	\[ \left| \widehat{S}_N-\mathrm{ES}_\alpha(X) \right| \stackrel{\P}{\longrightarrow} 0 \quad\text{as } N\to\infty,\]
	where $\stackrel{\P}{\to} 0$ denotes convergence in probability.
\end{corollary}

\section{On the optimality}
\label{sec:onoptimality}
In this section, we rigorously prove the optimality claims made in the introduction.
To that end, let us start with some preliminaries.
We shall often make use of scaled Bernoulli distributions:
For $x>0$, write 
\[X\sim {\rm B}(p,x) \quad\text{if } \P(X=0)=1-p \text{ and } \P(X=x)=p.\]
In particular, if  $X\sim {\rm B}(p,x)$, then
\begin{align}
\label{eq:Binom.F}
F_X = (1-p) \eins_{[0,x)} + \eins_{[x,\infty)} \quad\text{and}\quad 
F_X^{-1}= -\infty \eins_{\{0\}} + 0\eins_{(0,1-p]}+ x\eins_{(1-p,1]}.
\end{align} 
The reason to consider the Bernoulli distribution steams from the fact that its empirical distribution remains a Bernoulli distribution (with a random success parameter), and that explicit calculations are simpler.

\begin{lemma}
\label{lem:ES.bernoulli}
	Let $p\in[0,1]$, $x>0$, and $X\sim {\rm B}(p,x)$.
	Then
	\begin{align*}
	\ES_\alpha(X)
	&=x \min\left\{1,\frac{p}{\alpha}\right\},\\
	\sigma^2_{\ES_\alpha(X)}
	&= \begin{cases}
	\frac{x^2(p-p^2)}{\alpha^2} 
	&\text{if } p\leq \alpha, \\
	0 
	&\text{else}.
	\end{cases}
	\end{align*}
\end{lemma}
\begin{proof}
	By definition of the expected shortfall and  by \eqref{eq:Binom.F},
	\[\ES_\alpha(X)
	=\frac{1}{\alpha} \int_{1-\alpha}^1 x \eins_{(1-p,1]}(u) \,du
	=
	\begin{cases}
	\frac{1}{\alpha} x\alpha &\text{if } p\geq\alpha,\\
	\frac{1}{\alpha} xp &\text{if } p<\alpha
	\end{cases}\]
	from which the first claim readily follows.
	
	As for the second claim, assume first that $p\leq \alpha$.
	Then  $F_X^{-1}(1-\alpha)=0$ (see \eqref{eq:Binom.F}) and therefore
	\[\sigma^2_{\ES_\alpha(X)}=\frac{1}{\alpha^2}\int_0^\infty \int_0^\infty F_X(\min\{t, s\}) - F_X(t)F_X(s)\, dtds.\]
	Moreover, note that for $t\geq x$, $F_X(t)=1$, hence, if either $t\geq x$ or $s\geq x$, then $F_X(\min\{t, s\}) - F_X(t)F_X(s)=0$.
	On the other hand, for $t,s\in[0,x)$, the latter term is equal to $1-p-(1-p)^2=p-p^2$ from which the claim follows.	

	The proof in the case $p>\alpha$ follows because  $F^{-1}_X(1-\alpha)=x$ (see \eqref{eq:Binom.F}) and hence
	\begin{align*}
	\sigma^2_{\ES_\alpha(X)}
	&=\frac{1}{\alpha^2}\int_x^\infty \int_x^\infty F_X(\min\{t, s\}) - F_X(t)F_X(s)\, dtds \\
	&= \frac{1}{\alpha^2} \int_{x}^\infty \int_x^\infty 1 - 1 \,dtds = 0.
	\qedhere
	\end{align*}
\end{proof}

In what follows, we will often use the following standard estimate, the proof of which we provide for completeness.

\begin{lemma}
\label{lem:prob.exists}
	Let $N\geq 1$, let $p\in[0,\frac{1}{N}]$, and let  $(\mathcal{A}_i)_{i=1}^N$ be independent events that satisfy $\P(\mathcal{A}_i) \geq p$ for all $i=1,\dots,N$.
	Then
	\[ \P\left( \bigcup\nolimits_{i=1}^N \mathcal{A}_i \right)
	\geq \frac{1}{4} Np. \]
\end{lemma}
\begin{proof}
	Clearly
	\[\P\left( \bigcup\nolimits_{i=1}^N \mathcal{A}_i \right)
	=1-\P\left( \bigcap\nolimits_{i=1}^N \mathcal{A}_i^c \right)
	\geq 1-(1-p)^N.\]
	Moreover, using that  $\log(1-x) \leq - x$ for $x\in(0,1]$, 
	\[ 1-(1-p)^N
	=1-\exp(N \log(1-p))
	\geq 1- \exp\left(-Np \right).\]
	Finally, since $\exp(-x)\geq 1-\frac{1}{4}x$ for $x\in[0,1]$, we conclude that
	\[\P\left( \bigcup\nolimits_{i=1}^N \mathcal{A}_i \right)
	\geq 1- \exp\left(- Np \right)
	\geq \frac{1}{4} Np,\]
	as claimed. 
\end{proof}

\subsection{The minimax rate}
\label{sec:minimax}

Let us recall that $\widehat{T}_N$ is the plug-in estimator for $\ES_\alpha(X)$.
In this section we prove that the linear rate $\mathcal{O}(\frac{1}{N\varepsilon^2})$ for the rate of convergence of $\P( | \widehat{T}_N - \mathrm{ES}_\alpha(X) | \geq \varepsilon \sigma_{\ES_\alpha(X)})$ is optimal in a minimax sense.
The alternative result using only Pareto-distributions is presented in Section \ref{sec:Pareto}.

\begin{proposition}
\label{prop:heavy.tails.example}
	Let $\alpha\leq \frac{1}{10}$, let  $\varepsilon\leq \frac{1}{2}\sqrt\alpha$ and assume that $N\varepsilon^2\geq 5$.
	Then there exists a random variable $X$ such that $\sigma_{\ES_\alpha(X)}=1$, $F_X^{-1}$ is $1$-Lipschitz on $[1-2\alpha,1-\frac{1}{2} \alpha]$ and 
	\[ \P\left( \left| \widehat{T}_N - \mathrm{ES}_\alpha(X) \right| \geq \varepsilon \sigma_{\ES_\alpha(X)} \right)
	\geq \frac{1}{16 N\varepsilon^2}.\]
\end{proposition}
\begin{proof}
	For shorthand notation, put $\delta=2\varepsilon$. 
	Fix $N\geq 1$ and set 
	\[	p:= \frac{ 1 }{ \delta^2 N^2 } 
	\quad\text{and}\quad
	x:= \frac{\alpha}{\sqrt{p-p^2}}. \]
	Let $X\sim {\rm B}(p,x)$ so that  $\sigma^2_{\ES_\alpha(X)}=1$ by Lemma \ref{lem:ES.bernoulli}.
	
	Let us start by collecting some trivial estimates on $p$.
	By assumption $N\delta^2\geq 20$ and $\delta^2\leq \alpha$, hence $N\geq \frac{20}{\alpha}$ and thus $p\leq\frac{1}{N}\leq\frac{\alpha}{20}\leq \frac{1}{100}$.
	Moreover, since $N\delta^2\geq 20$, $\sqrt{p}\leq \frac{1}{20}\delta$.
	We will use these facts later.
	In particular, let us note that since $p\leq\frac{\alpha}{2}$,   $F_X^{-1}$ is constant on $(1-2\alpha,1-\alpha/2)$ (see \eqref{eq:Binom.F}) and thus also $1$-Lipschitz.
	
	Define the empirical success probability $\widehat{p}_N= \frac{1}{N}\sum_{i=1}^N \eins_{X_i=x}$ so that the  empirical distribution $\widehat{F}_N$ of $X$ is equal to that of ${\rm B}(\widehat{p}_N,x)$.
	Consider the event
	\[ \mathcal{A}:=\left\{  \widehat{p}_N \geq \frac{1}{N} \right\}
	=\bigcup_{i=1}^N \{ X_i = x\}.\]
	Since $\P( X_i = x) = p $ and $p\leq \frac{1}{N}$, Lemma  \ref{lem:prob.exists} guarantees that 
	\[\P(\mathcal{A})
	\geq\frac{1}{4}Np 
	= \frac{1}{16 N\varepsilon^2}.\]
	It remains to show that in the event $\mathcal{A}$, $| \widehat{T}_N - \ES_\alpha(X) |\geq \frac{1}{2}\delta= \varepsilon\sigma_{\ES_\alpha(X)}$.
	
	To that end, we first compute $\ES_\alpha(X)$.
	By Lemma \ref{lem:ES.bernoulli} and since  $p\leq \alpha$, 
	\[\ES_\alpha(X)
	 = \frac{xp}{\alpha}
	= \frac{p}{\sqrt{p-p^2}}
	\leq  \frac{11}{10}\sqrt{p} , \]
	where the inequality follows because $p\leq\frac{1}{100}$ and therefore $\sqrt{p-p^2}\geq \frac{10}{11}\sqrt{p}$.
	
	On the other hand, in the event $\mathcal{A}$ we have $\widehat{p}_N\geq \frac{1}{N}$.
	Since $N\alpha\geq N\varepsilon^2\geq 1$, it follows again from Lemma \ref{lem:ES.bernoulli} that
	\[  \widehat{T}_N 
	= x\min\left\{ 1, \frac{\widehat{p}_N}{\alpha}   \right\}
	\geq x \min\left\{1,  \frac{1}{N\alpha}   \right\}
	 = \frac{x}{N\alpha}. \]
	Moreover, by definition of $p$ and $x$, 
	\[  \frac{x}{N\alpha} 
	=\frac{1}{N\sqrt{p-p^2}}
	\geq \frac{1}{N\sqrt{p}}
	=\delta.\]
	Therefore, using that $\sqrt{p}\leq\frac{1}{20}\delta$, in the event  $\mathcal{A}$, 
	\[    \widehat{T}_N -\ES_\alpha(X)
	\geq \delta - \frac{11 }{10} \sqrt{p}
	\geq \frac{\delta}{2} \]
	which completes the proof.
	\end{proof}

\subsection{On the continuity of the quantile functions}
\label{sec:continuity.inverse}

In this section we make rigorous the claim in the introduction that if $u\mapsto \VAR_u(X)$ is $L$-Lipschitz continuous, then the typical behavior of the estimation error scales like
\[ \left| \widehat{T}_N -  \ES_\alpha(X) \right|
=\mathcal{O}\left( \frac{\sigma_{\ES_\alpha(X)}}{ \sqrt N} \right) + \mathcal{O}\left( \frac{L}{ N} \right)  \]
in general.
In the next proposition, we focus on the (perhaps more surprising) $ \mathcal{O}( \frac{L}{ N} ) $-term; the $\mathcal{O}( \frac{\sigma_{\ES_\alpha(X)}}{\sqrt N} )$-term will be dealt with in Proposition \ref{prop:CLT.sharp}.

\begin{proposition}
\label{prop:inverse.continuity} 
	There are absolute constants $c_0,c_1,c_2, c_3>0$ such that the following holds.
	Let $L\geq 1$ and set  $\varepsilon\leq \min\{ c_0\sqrt{\alpha}, \frac{1}{L} \}$.
	Then there exists a random variable $X$ which satisfies
	\begin{enumerate}[label = (\alph*), leftmargin=2em]
	\item $|X|\leq 1$ almost surly,
	\item $F_X^{-1}$ is $L$-Lipschitz on $(0,1)$,
	\item $\sigma^2_{\ES_\alpha(X)}=0$,
	\end{enumerate}	 
	and for every $N\geq \frac{c_1}{\varepsilon^2}$, we have that
	\[ \P\left( \left| \widehat{T}_N -  \ES_\alpha(X) \right|
	\geq c_2 \varepsilon^2 L  \right)
	\geq \exp\left( - c_3 N \varepsilon^2 \right).\]
\end{proposition}

\begin{remark}
	In the setting and notation of Proposition \ref{prop:inverse.continuity}: 
	Suppose that $L\geq  \frac{1}{c_0\sqrt\alpha}$, that $N\geq c_1 L^2$, and set $\varepsilon= \sqrt{c_1/N}$.
	Then, with probability at least $\exp(-c_1 c_3)$,
	\[ \left| \widehat{T}_N -  \ES_\alpha(X) \right|
	\geq c_1c_2 \frac{ L}{ N};\]
	and in particular $\mathbb{E}[| \widehat{T}_N -  \ES_\alpha(X) |] \geq c \frac{L}{N}$ for $c=c_1c_2\exp(-c_1 c_3)>0$.
\end{remark}

\begin{proof}[Proof of Proposition \ref{prop:inverse.continuity}]
	{\em Step 1:}
	Let  $\delta:=\frac{1}{40}\varepsilon\sqrt\alpha$, put $x_0: = -\delta L$, and set $X$ to have the following distribution:
	\[ 
	F_X(t) 
	=\begin{cases}
	0 & \text{if } t<x_0 ,\\
	1-\alpha - \delta + \frac{x_0-t}{x_0}\delta & \text{if } t\in[x_0,0) , \\
	1 & \text{if } t\geq 0	.
	\end{cases} \]
	In other words, $\P(X=x_0)=1-\alpha-\delta$, $\P(X=0)=\alpha$ and for $A\subset (x_0,0)$ we have that   $\P(X\in A) = \frac{\delta}{|x_0|} \,{\rm Leb}(A)$.
	Clearly
	\[F_X^{-1}(u)
	=\begin{cases}
	x_0 & \text{if } u\in(0,1-\alpha-\delta) ,\\
	x_0 - \frac{x_0}{\delta}(u-(1-\alpha-\delta))  & \text{if } u\in [1-\alpha-\delta, 1-\alpha) , \\
	0 & \text{if } u\in[1-\alpha,1),
	\end{cases} \]
	and in particular $F_X^{-1}$ is Lipschitz continuous on $(0,1)$ with constant $L=\frac{|x_0|}{\delta}$.
	Moreover, since $\varepsilon\leq \frac{1}{L}$, it follows that $|x_0|\leq \varepsilon L \leq 1$, thus $|X|\leq 1$ almost surely.

	Finally, since  $F_X^{-1}(u)=0$ for $u\in(1-\alpha,1)$, it follows from the definitions of $\ES_\alpha$ and $\sigma^2_{\ES_\alpha}$ that   $\ES_\alpha(X)=\sigma^2_{\ES_\alpha(X)}=0$.
	
	\vspace{0.5em}
\noindent
{\em Step 2:}
	We claim that there is an absolute constant $c$ such that
	\[ \P\left(\widehat{F}_N(x_0) \geq 1-\alpha +\delta\right)
	\geq c \exp\left(-\frac{1}{2} \varepsilon^2 N \right).\]	
	
	To that end, we shall use the following Binomial anti-concentration result due to Feller.
	It is a consequence of  \cite[Theorem 1]{feller1943generalization} that if $(Z_i)_{i=1}^N$ are i.i.d.\ random variables satisfying  that $|Z_i|\leq 1$ and  $N\geq \frac{400}{\V[Z]}$, then for every $\lambda \in ( 0, \frac{\V[Z]}{100})$, 
	\begin{align}
	\label{eq:Feller}
	 \P\left( \frac{1}{N}\sum_{i=1}^N Z_i \geq \E[Z]+\lambda \right)
	\geq c\exp\left(-\frac{\lambda^2N}{3  \V[Z]} \right).
	\end{align}
	We apply Feller's result to $Z_i=\eins_{(-\infty,x_0]}(X_i)$.
	If we assume that $\varepsilon \leq \sqrt{\alpha}$, then
	\[ \V[Z]=(1-\alpha-\delta)(\alpha+\delta)\in\left[\frac{\alpha}{2},2\alpha\right]\]
	by the definition of $\delta=\frac{1}{40}\varepsilon\sqrt\alpha$.
	Moreover, if $\varepsilon\leq\frac{1}{40}\sqrt{\alpha}$,  $\lambda:=2\delta$ satisfies that $\lambda \leq \frac{1}{100}\V[Z]$.
	Therefore, by \eqref{eq:Feller},
	\[ \P\left( \widehat{F}_N(x_0) \geq F(x_0)+2\delta \right)
	\geq c\exp\left(- \frac{ 8\delta^2 N}{ 3\alpha} \right)
	= c\exp\left(- \tilde{c}\varepsilon^2 N \right)\]
	for $\tilde{c}:=\frac{8}{3 \cdot 40^2 }$.
	Finally, if $\varepsilon^2 N\geq \frac{2}{\tilde{c}}\log(\frac{1}{c})$, then $c\exp\left(- \tilde{c}\varepsilon^2 N \right)\geq \exp(-2\tilde{c}\varepsilon^2N)$.
	
	\vspace{0.5em}
\noindent
{\em Step 3:}
	Fix a realization in the high probability event of step 2, i.e., such that $\widehat{F}_N(x_0) \geq 1-\alpha +\delta$.
	Then, by definition of the inverse function, 
	\[\widehat{F}_N^{-1}(u)\leq x_0 \quad\text{for all } u<1-\alpha+\delta.\]
	Moreover, clearly $\widehat{F}_N^{-1}(\cdot)\leq 0$, and therefore
	\[ \widehat{T}_N = \frac{1}{\alpha} \int_{1-\alpha}^1 \widehat{F}_N^{-1}(u) \,du \leq
	\frac{1}{\alpha} \int_{1-\alpha}^{1-\alpha+\delta} \widehat{F}_N^{-1}(u) \,du 
	\leq \frac{x_0\delta}{\alpha}.\]
	In particular, since $\ES_\alpha(X)=0$, 
	\[|\widehat{T}_N-\ES_\alpha(X)|
	\geq \frac{|x_0|\delta}{\alpha}
	= \frac{L\delta^2}{\alpha}
	=\frac{L\varepsilon^2}{1600},\]
	which concludes the claim.
\end{proof}

\subsection{No estimator can outperform the CLT-rates}	\label{sec:no.estimtor.better.than.clt}

We prove the claim from the introduction that 
\[ \P\left( \left| \widehat{T}_N - \ES_\alpha(X) \right| \gtrsim \varepsilon \sigma_{\ES_\alpha(X)} \right) 
	\geq \exp\left(- N\varepsilon^2\right).\]
As it happens, such a lower bound is valid not only for the plug-in estimator $\widehat{T}_N$, but for \emph{any} estimator:

\begin{proposition}
\label{prop:CLT.sharp}
	Let $\varepsilon\leq \frac{1}{2}\sqrt\alpha$ and let $\widehat{R}_N$ be any estimator.
	Then there exists a random variable $X$ for which $\sigma_{\ES_\alpha(X)}=1$ and $F_X^{-1}$ is $1$-Lipschitz on $[1-2\alpha,1-\frac{1}{2}\alpha]$ and 
	\[ \P\left( \left| \widehat{R}_N - \ES_\alpha(X) \right| \geq  \frac{1}{10} \varepsilon \sigma_{\ES_\alpha(X)} \right) 
	\geq \exp\left(- N\varepsilon^2\right).\]
\end{proposition}
\begin{proof}
	Set
	\[(p,x):=\left( \varepsilon^2, \frac{\alpha}{\sqrt{\varepsilon^2-\varepsilon^4}}\right),
	\quad\text{and}\quad
	(q,y):=\left( \frac{\varepsilon^2}{2}, \frac{\alpha}{\sqrt{\varepsilon^2/2-\varepsilon^4/4}} \right).\]
	Moreover, let  $X\sim {\rm B}(p,x)$ and $Y\sim {\rm B}(q,y)$.
	It follows from Lemma \ref{lem:ES.bernoulli} that $\sigma^2_{\ES_\alpha(X)}=\sigma^2_{\ES_\alpha(Y)}=1$.
	By the assumption that $\varepsilon\leq \frac{1}{2}\sqrt\alpha$, we have that $p,q\leq \frac{1}{4}\alpha$.
	In particular, the two functions $F_X^{-1}$ and $F_Y^{-1}$ are constant on $(1-2\alpha, 1-\frac{1}{2}\alpha)$, and thus $1$-Lipschitz.
	Moreover, by  Lemma \ref{lem:ES.bernoulli},
	\[  \ES_\alpha(X)
	=\min\left\{ \frac{xp}{\alpha} ,x  \right\}
	=\frac{xp}{\alpha} 
	=\frac{p}{\sqrt{p-p^2}}
	\geq \frac{p}{\sqrt p}
 	= \varepsilon.\]
	In a similar manner, 
	\[  \ES_\alpha(Y)
	=\min\left\{ \frac{yq}{\alpha} ,y \right\}
	=\frac{q}{\sqrt{q-q^2}}\]
	and since $q= \frac{\varepsilon^2}{2}\leq \frac{1}{8}$ it holds that $q-q^2\geq \frac{7}{8}q$ and therefore
	\[\ES_\alpha(X) - \ES_\alpha(Y)
	\geq \varepsilon - \sqrt\frac{8}{7} \sqrt{q}
	\geq \frac{1}{5}\varepsilon.\]
 	
 	We couple the random variables $X$ and $Y$ such that $\P(X\neq Y)=\frac{1}{2}\varepsilon^2$, and let $(X_i,Y_i)_{i=1}^N$ be an i.i.d.\, sample of the random vector $(X,Y)$.
	Then, setting
	\[ \mathcal{A}:=\left\{  X_i=Y_i \text{ for all } i=1,\dots, N\right\},
	\]
	it follows that
	\[ \P(\mathcal{A})=\left(1-\frac{\varepsilon^2}{2}\right)^N
	=\exp\left( N\log\left(1-\frac{\varepsilon^2}{2}\right)\right)
	\geq \exp(-N\varepsilon^2),\]
	where the inequality holds because $\varepsilon^2\leq\frac{1}{4}$  and $\log(1-x)\geq -2x$ for $x\in[0,\frac{1}{8}]$.

	Now let $\widehat{R}_N\colon \R^N\to\R$ be any estimator.
	Then, in the event $\mathcal{A}$,
	\[ \widehat{R}_N\left( (X_i)_{i=1}^N \right)
	=\widehat{R}_N\left( (Y_i)_{i=1}^N \right).\]
	On the other hand, since $|\ES_\alpha(X)-\ES_\alpha(Y)|\geq \frac{1}{5}\varepsilon$, the (reverse) triangle inequality implies that it must either hold that 
	\[\left| \widehat{R}_N\left( (X_i)_{i=1}^N \right)-\ES_\alpha(X) \right|
	\geq  \frac{\varepsilon}{10}
	\text{ or }
	\left| \widehat{R}_N\left( (Y_i)_{i=1}^N \right)-\ES_\alpha(Y) \right|
	\geq  \frac{\varepsilon}{10},\]
	which completes the proof.
\end{proof}

\subsection{The level of adversarial corruption}

Suppose that an adversary can modify at most $K$ out of the $N$ samples, where $K\lesssim \alpha N$.
In Theorem \ref{thm:adversarial} (setting $\varepsilon \sim \sqrt{K/N}$) we have shown that  the estimator $\widehat{S}_N$ satisfies
\[ \P \left( \left|\widehat{S}_N-\ES_\alpha(X) \right|\gtrsim \sqrt{\frac{K}{N}} \sigma_{\ES_\alpha(X)} +  \frac{K}{N} L \right)
\leq \exp(-cK).\]
In the next proposition we focus on the $\sqrt{\frac{K}{N}} $ term and show that the given dependence is minimax optimal.

\begin{proposition}
\label{prop:adversiarl}
	Let $\alpha\leq \frac{2}{10}$, let $N\geq 10$, let $K\leq\frac{1}{2}\alpha N$, and let  $\widehat{R}_N$ be any estimator.
	Then there exists a random variable $X$ for which $\sigma_{\ES_\alpha(X)}=1$ and $F_X^{-1}$ is $1$-Lipschitz on $[1-2\alpha,1-\frac{1}{2}\alpha]$ and the following holds:
	with probability at least $1-2\exp(- \frac{K}{4})$,  an adversary can modify $K$ of the $N$ samples in a way such that 
	\[ \left| \widehat{R}_N - \ES_\alpha(X) \right| \geq  \frac{1}{20} \sqrt\frac{K}{N} \sigma_{\ES_\alpha(X)}  .\]
\end{proposition}
\begin{proof}
	Set $\varepsilon:=\sqrt{\frac{K}{4N}}$.
	We use the same notation as the proof of Proposition \ref{prop:CLT.sharp}; in particular $X$ and $Y$ are the Bernoulli random variables that appear in the proof of Proposition \ref{prop:CLT.sharp} so that  $\P(X\neq Y)=\frac{1}{2}\varepsilon^2$.
	
	Set
	\[U:=|\{i \in \{1,\dots,N\}: X_i\neq Y_i\}|\]
	so that  $\E[U]=\frac{1}{2}\varepsilon^2N$.
	It follows from Bernstein's inequality that for every $\lambda\geq 0$, with probability at least $1-2\exp(-\lambda)$,
	\[  \frac{U}{N} - \frac{\varepsilon^2}{2}
	= \frac{1}{N} \sum_{i=1}^N \left( \eins_{X_i\neq Y_i} - \E[\eins_{X_i\neq Y_i}] \right)
	\leq \sqrt \frac{2 \V[ \eins_{X\neq Y} ] \lambda }{N} + \frac{2\lambda}{3N} .\]
	In particular, noting that $\V[ \eins_{X\neq Y} ]\leq 2\varepsilon^2$ and applying Bernstein's inequality to $\lambda=\varepsilon^2 N$, it follows that with probability $1-2\exp(-\varepsilon^2 N)$, one has that $U\leq 4\varepsilon^2 N=K$.
	
	Further, on the event $\{U \leq  K\}$, the adversary can modify the sample of $X$ in a way that it is equal to the sample of $Y$.
	Finally,   we have seen in the proof of Proposition \ref{prop:CLT.sharp} that $|\ES_\alpha(X)-\ES_\alpha(Y)|\geq \frac{1}{5}\varepsilon$, from which the statement readily follows.
\end{proof}
	
\section{Explicit computations in case of the Pareto distribution}
\label{sec:Pareto}

Recall that $X$ is said to have a Pareto distribution with parameters $\lambda>0$ and $x_0\in\mathbb{R}$ if $\P(X\leq t)=1- (\frac{x_0}{t})^\lambda 1_{ [x_0,\infty) }(t)$ for $t\in\mathbb{R}$.
Let us first prove the finite sample rate mentioned in the introduction.

\begin{lemma}\label{lem:paretolower}
	Assume that $X$ has a Pareto distribution with parameters $\lambda>2$ and $x_0>0$.
	Then, for all $N \geq \frac{1}{\alpha}(\frac{\lambda-1}{\lambda})^\frac{\lambda}{\lambda-1}$, it holds that
	\[ 
	\P\left( \left| \widehat{T}_N-  \ES_{\alpha}(X) \right| \geq \varepsilon \sigma_{\ES_\alpha} \right)
	\geq \frac{1}{2} \left(\frac{x_0}{\alpha \left(\ES_{\alpha}(X) + \varepsilon \sigma_{\ES_\alpha}\right)}\right)^{\lambda} \frac{1}{N^{\lambda-1}}.
	\]
\end{lemma}
\begin{proof}
	Observe that for all $t >0$,
	\begin{align*}
	\P\left( \widehat{T}_N \geq t \right)
	&\geq \P\left( X_i\geq  \alpha N t \text{ for some } 1\leq i\leq N \right)
	=:(\ast).
	\end{align*}
	Moreover,  by the definition of the Pareto distribution, $\P( X\geq \alpha N t ) =(\frac{ x_0 }{\alpha t N} )^\lambda$.
	If $t>0$ is such that  $N^{\lambda-1} \geq (\frac{x_0}{\alpha t})^\lambda$, then $(\frac{ x_0 }{\alpha t N} )^\lambda\leq \frac{1}{N}$ and we may apply Lemma \ref{lem:prob.exists}, which yields
	\[ (\ast) \geq \frac{1}{2} \frac{ x_0^\lambda }{\alpha^\lambda t^\lambda N^{\lambda-1}}. \]
	Therefore, the proof is completed if we may apply the previous estimate to $t = \ES_{\alpha}(X) + \varepsilon \sigma_{\ES_\alpha}$ -- in other words, if this choice of $t$ satisfies that  $N^{\lambda-1} \geq \left(\frac{x_0}{\alpha t}\right)^{\lambda}$.
	To that end, it suffices to note that
	\[ \ES_{\alpha}(X) = \frac{x_0 \lambda}{\alpha^{1/\lambda}(\lambda-1)},
	\]
	see, e.g, \cite[Proposition 2]{norton2021calculating}; thus the desired estimate $N^{\lambda-1} \geq \left(\frac{x_0}{\alpha t}\right)^{\lambda}$ follows from the assumption that $N^{\lambda-1} \geq \frac{1}{\alpha^{\lambda-1}}\left(\frac{\lambda -1 }{\lambda}\right)^{\lambda}$.
\end{proof}

\begin{lemma}\label{lem:paretosigma}
	Assume that $X$ has a Pareto distribution with parameters $\lambda > 2$ and $x_0>0$.
	Further set  $r := \frac{\lambda+1}{\lambda}$ and let $\alpha\leq \frac{1}{2}$.
	Then the following hold.
	\begin{enumerate}[label = (\roman*), leftmargin=2em]
	\item $F_X^{-1}$ is Lipschitz continuous on $[1-2\alpha,1-\alpha/2]$ with constant
	\[ L=\frac{x_02^r}{\lambda \alpha^r}.\]
	\item 
	We have that 
	\[
	\frac{x_0^2 \alpha^{1-2 r}}{2\lambda^2(3-2r)}
	\leq 
	\sigma_{\ES_\alpha(X)}^2 
	\leq \frac{2x_0^2 \alpha^{1-2 r}}{\lambda^2(r-1)(3-2r)}.
	\]
	\end{enumerate}
	\end{lemma}

\begin{remark}
	Recall that the error for the estimation of the expected shortfall risk scales like 
	\begin{align}
	\label{eq:final.remark}
	|\widehat{S}_N- \ES_\alpha(X)|
	=\mathcal{O}(\varepsilon \sigma_{\ES_\alpha(X)}) + \mathcal{O}(\varepsilon^2 L),
	\end{align}
	see Theorem \ref{thm:ES}.
	Hence, the ratio $\frac{\sigma_{\ES_\alpha}}{L}$ is the threshold such that for all $\varepsilon> \frac{\sigma_{\ES_\alpha}}{L}$ the second error term in \eqref{eq:final.remark} is the dominant one, and for all $\varepsilon\leq  \frac{\sigma_{\ES_\alpha}}{L}$  the first term in \eqref{eq:final.remark} is the dominant one; see also the discussion proceeding Corollary \ref{cor:regimeswitch}.
	
	Lemma \ref{lem:paretosigma} shows that for the Pareto distribution, $\frac{\sigma_{\ES_\alpha}}{L} \geq \frac{ \sqrt{\alpha} }{ 2^r \sqrt{2}} \geq \frac{1}{4} \sqrt{\alpha}$, i.e.\ for all $\varepsilon\leq \frac{1}{4} \sqrt{\alpha}$, the $\mathcal{O}(\varepsilon\sigma_{\ES_\alpha(X)})$ term  in \eqref{eq:final.remark} is the dominant one.
\end{remark}
	
\begin{proof}[Proof of Lemma \ref{lem:paretosigma}]
	We start by proving the first claim.
	To that end, write $F=F_X$, denote by $f=F'$ the density of $X$, and note that
	\begin{align}
	\label{eq:F.inverse.diff}
	\frac{d}{du} F^{-1}(u) 
	=\frac{1}{f(F^{-1}(u))}.
	\end{align}
	It is straightforward to verify that  $f(t) = \lambda x_0^\lambda  t^{-\lambda -1}$ for $t\geq x_0$ and  that $F^{-1}(u) =  x_0(1-u)^{-1/\lambda}$, from which the first claim readily follows.
	
	As for the second claim, we assume that $x_0=1$ for simpler notation.
	By definition of $	\sigma^2_{\ES_\alpha}$ and by a change of variables
	\begin{align*}
	\sigma^2_{\ES_\alpha(X)}
	&=\frac{1}{\alpha^2} \int_{F^{-1}(1-\alpha)}^\infty\int_{F^{-1}(1-\alpha)}^\infty F(\min\{t, s\}) - F(t)F(s)\,dtds\\
	&=\frac{1}{\alpha^2} \int_{1-\alpha}^1 \int_{1-\alpha}^1  \frac{ \min\{u, v\} -uv}{\lambda^2 (1-u)^{\frac{\lambda+1}{\lambda}}  (1-v)^{\frac{\lambda+1}{\lambda}} } \, du dv.
	\end{align*}
	Indeed,  we used the substitution $u=F(t)$  (hence  $t=F^{-1}(u)$ and $dt=\frac{du}{f(F^{-1}(u))}$) and similarly $v=F(s)$.
	Recall that $r=\frac{\lambda+1}{\lambda}$ and set, for every $v\in[1-\alpha,1)$,
	\[ 
	A_v:=\int_{1-\alpha}^v  \frac{ u(1-v)}{(1-u)^{r}  } \, du 
	\quad\text{and}
	\quad 
	B_v:=\int_{v}^1  \frac{ v(1-u)}{(1-u)^{r}  } \, du . \]
	Then we have that
	\begin{align}
	\label{eq:sigma.pareto.split}
	\sigma^2_{\ES_\alpha(X)}=\frac{1}{\alpha^2 \lambda^2} \int_{1-\alpha}^1 \frac{A_v+B_v}{(1-v)^r} \,d v.
	\end{align}

	We start by proving the lower bound on $\sigma_{\ES_\alpha(X)}^2$.
	To that end, note that
	\[B_v=\int_v^1 \frac{v }{(1-u)^{r-1}  } \, du 
	= \frac{-v }{(2-r)(1-u)^{r-2}  }\Big|_v^1
	= \frac{v }{(2-r)(1-v)^{r-2}  }.\]
	Moreover, for every $v\in[1-\alpha,1]$ we have that $A_v\geq 0$  and $v\geq \frac{1}{2}$; thus by \eqref{eq:sigma.pareto.split},
	\[\sigma^2_{\ES_\alpha(X)}
	\geq \frac{1}{\alpha^2 \lambda^2} \int_{1-\alpha}^1 \frac{1}{2(2-r)(1-v)^{2r-2}} \,d v
	= \frac{1}{\alpha^2 \lambda^2 2(2-r) (3-2r)\alpha^{2r-3}}.\]
	The claimed lower bound on $\sigma_{\ES_\alpha(X)}^2$ follows by noting that $r\in(1,\frac{3}{2})$ since $\lambda>2$.
	The upper estimate follows from the same arguments, additionally noting that
	\[A_v
 	\leq  (1-v)\int_{1-\alpha}^v \frac{1}{(1-u)^r}\,du
	\leq \frac{1}{(r-1)(1-v)^{r-2}}
	= \frac{2-r}{(r-1)v} B_v.	\qedhere\]
\end{proof}

\section{Numerical experiments}
\label{sec:numerics}

This section aims to shortly extend the numerical illustration of the proposed estimator from the introduction in a few more examples. For computational reasons, we always choose $\alpha=0.1$ so that a significant proportion of the data lies in the tails and simulations are hence less costly. The proposed estimator always uses $\beta_1=0.5$ and $\beta_2 = 0.6$. This is chosen such that it works well for block sizes $m$ around $250$. As discussed in Remark \ref{rem:bias}, for smaller values of $m$, one would need to widen the interval $[\beta_1, \beta_2]$, while for larger values of $m$, one can narrow the interval. 

As a short disclaimer, we emphasize that if, e.g., $\alpha=0.01$ were used instead of $\alpha=0.1$, we expect that the presented numerical results would be roughly similar if instead ten times as much data were used. In terms of practical applicability in financial domains, this means we are certainly optimistic with respect to the amount of data which we assume is available. Thus the results in this section should mostly be regarded as indications of the potential of the proposed estimator in an idealized setting. Future work for more realistic regimes (e.g., lower $N$) may combine the key ideas of the proposed estimator with practical adjustments like partly overlapping intervals or parametric assumptions on the forms of the tails.

Regarding the results, Figure \ref{fig:regimepareto} aims to illustrate the intuition of $\widehat{S}_N$ given in Remark \ref{rem:bias}, in that it can harness the benefits of the median of blocks estimator ($\beta_1=\beta_2=0.5$) for large deviations and high confidence, and the benefits of the plug-in estimator for small deviations and moderate confidence. Figure \ref{fig:regimestudent} shows that this behavior occurs for other heavy tailed distributions as well, and not just Pareto distributions. 

In Figure \ref{fig:trimmed}, we showcase a different estimator from the literature tailored to non-parametric estimation in heavy-tailed settings, namely two versions of a trimmed estimator from \cite{hill2015expected}, one standard version and one de-biased version. While we see that the standard version satisfies desirable properties for large deviations ($\delta=1$), it is slightly worse for moderate deviations ($\delta=0.5$), and vice versa for the de-biased version. We note that while the simulation study in \cite{hill2015expected} uses trimming of the largest $[0.25 n^{2/3}]$ samples, we found this too extreme in our experiments and thus trimmed just $[0.25 n^{1/3}]$ samples, which lead to much better results.

Figure \ref{fig:nicedistributions} showcases the behavior for more light-tailed distributions like the normal or (slightly less so) the log-normal distributions. In both cases, we can see that the plug-in estimator performs very well, and the proposed estimator almost matches its performance, while the median of blocks estimator performs worse in these well-behaved cases.

Figure \ref{fig:corrupted} illustrates Theorem \ref{thm:adversarial} by showcasing the robustness of the proposed estimator against adversarial manipulation of the data, whereas the plug-in estimator is heavily distorted even by manipulating just three data points.

Finally, Figure \ref{fig:theorydoesnothold} illustrates the behavior of the proposed estimator in a setting with infinite variance $\sigma_{\ES_\alpha}^2=\infty$.
While the theoretical guarantees from Theorem \ref{thm:ES} fail,  Figure \ref{fig:theorydoesnothold} shows that even in this situation, our estimator performs well, and in particular much better than the plug-in estimator.

\begin{figure}[H]
	\begin{minipage}{0.5\textwidth}
		\includegraphics[width=1\textwidth]{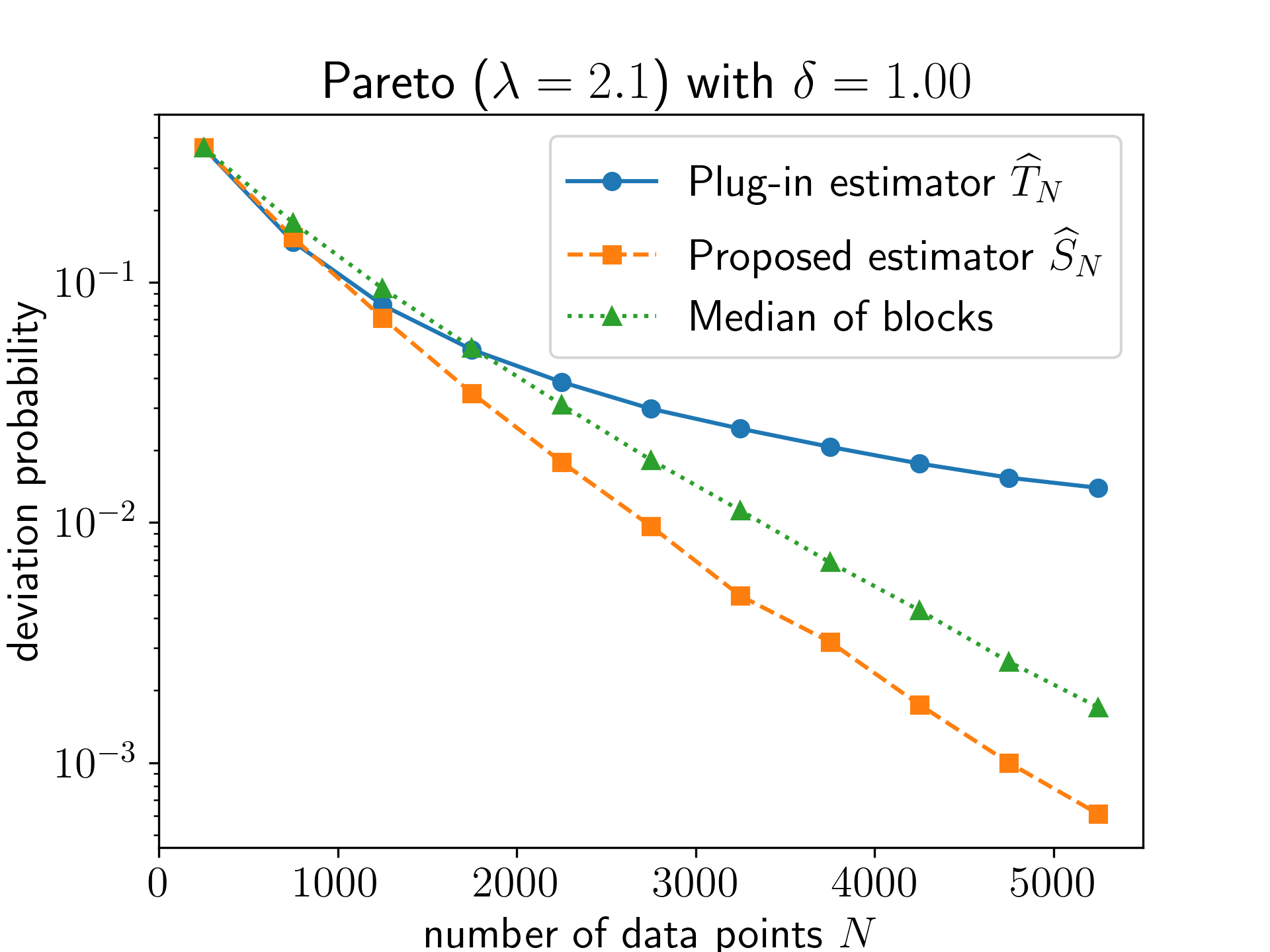}
	\end{minipage}%
	\begin{minipage}{0.5\textwidth}
		\includegraphics[width=1\textwidth]{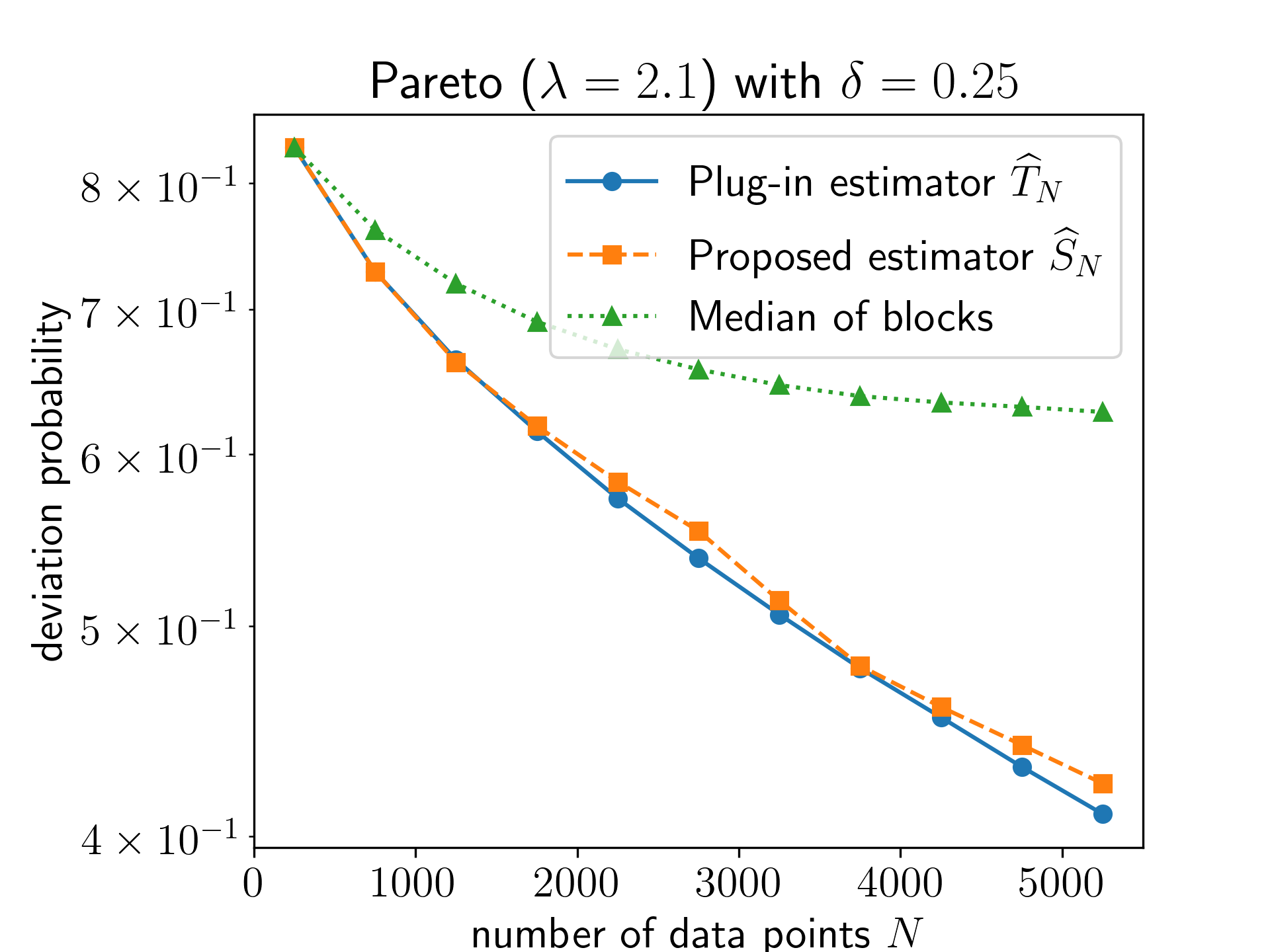}
	\end{minipage}
	\begin{minipage}{0.5\textwidth}
	\includegraphics[width=1\textwidth]{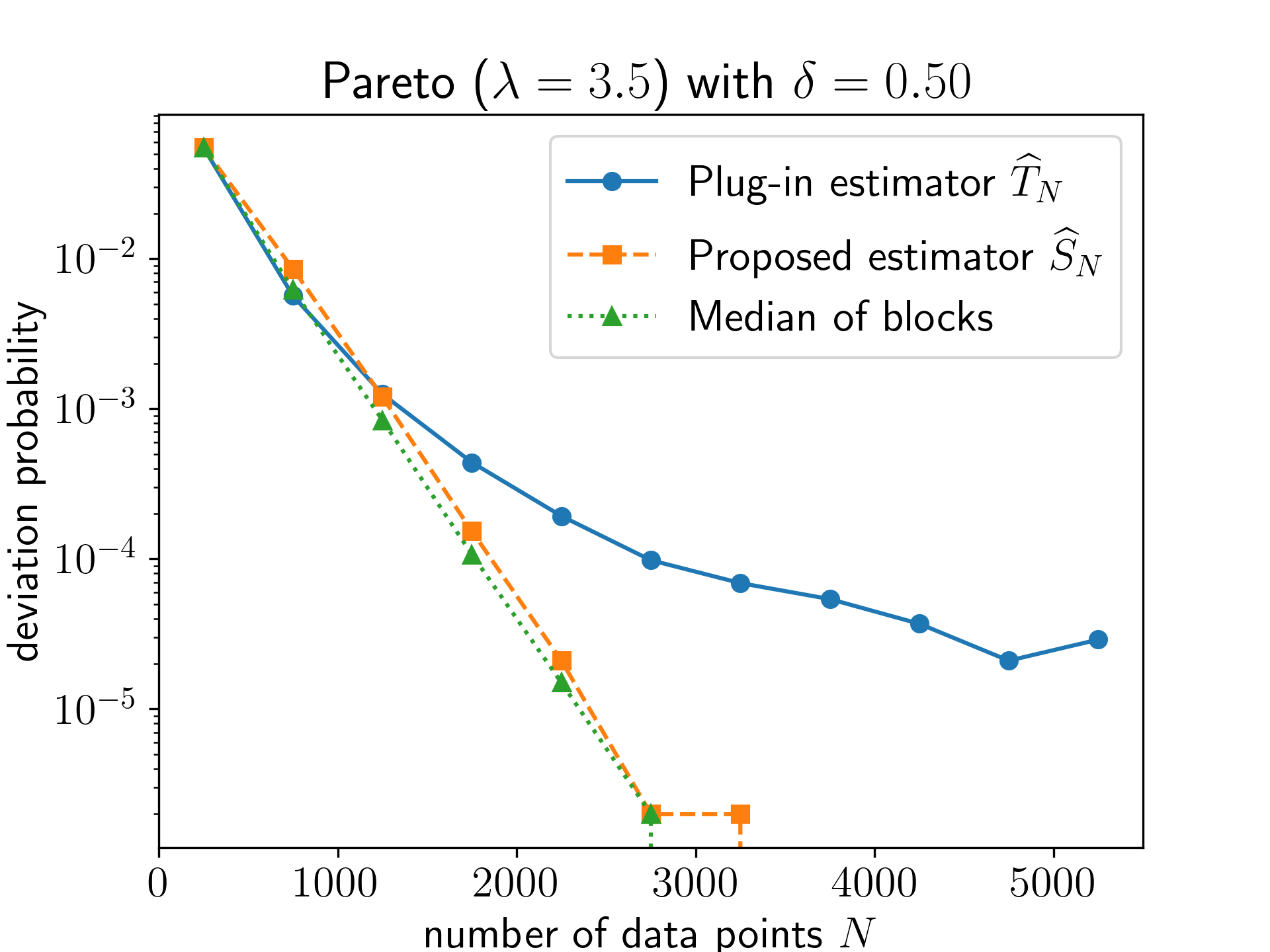}
	\end{minipage}%
	\begin{minipage}{0.5\textwidth}
		\includegraphics[width=1\textwidth]{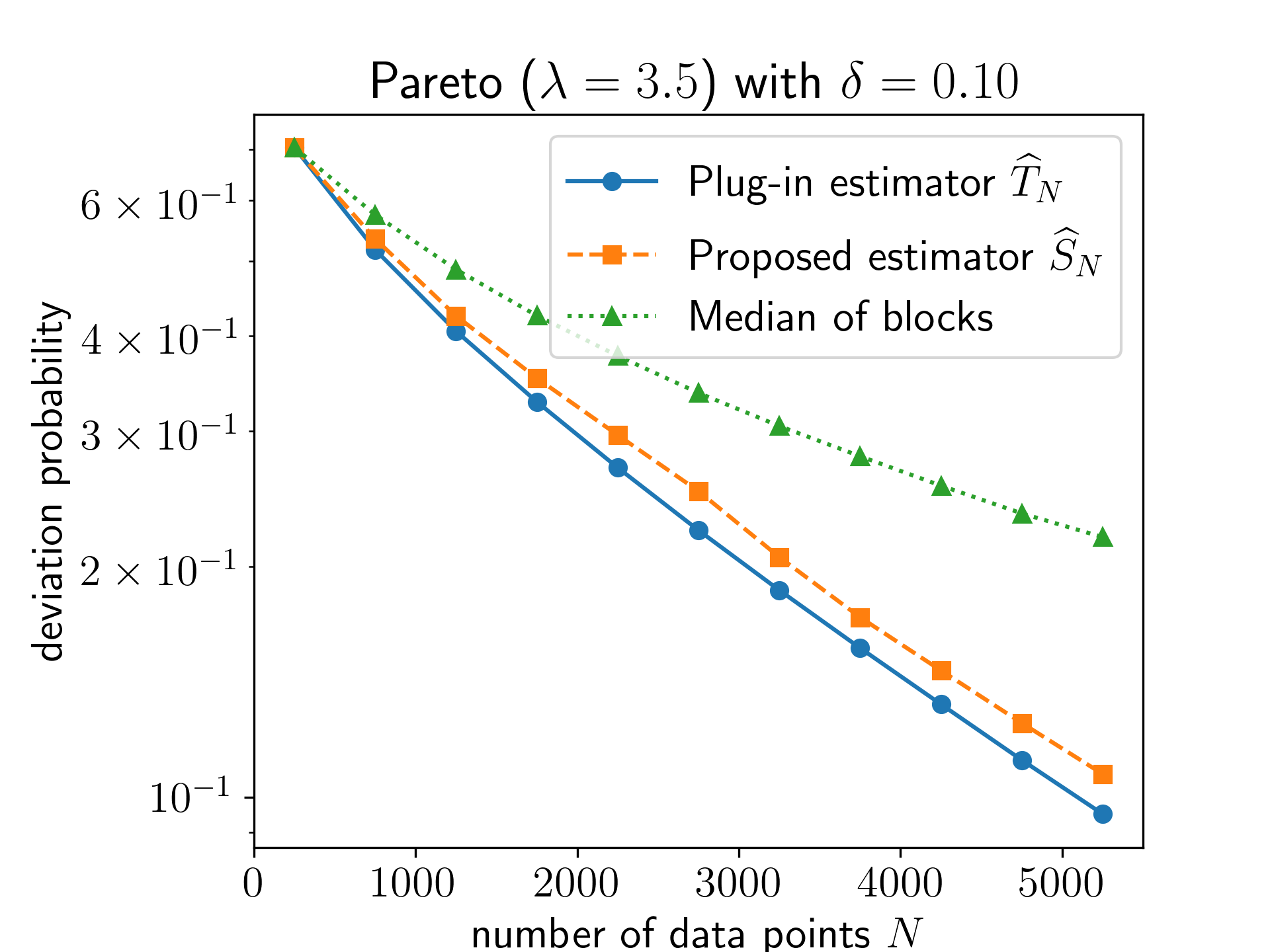}
	\end{minipage}%
\caption{The figures show $\mathbb{P}(|\widehat{R}_N - \ES_{\alpha}(X)| \geq \delta)$ (estimated using $10^6$ many experiments) for different estimators $\widehat{R}_N$, where $X$ is either Pareto distributed with $\lambda=2.1$ or $\lambda=3.5$. The proposed estimator and the median of blocks estimator use sub-intervals of size $m=250$. 
	The left hand side illustrates that the confidence bands for $\hat{S}_N$ and the median of blocks estimator behave like $\exp(-c_\delta N)$ as shown in Theorem \ref{thm:ES}, which is in contrast to the rate $\tilde{c}_\delta N^{-(\lambda - 1)}$ for $\widehat{T}_N$ as shown in Lemma \ref{lem:paretolower}.
	The figure further illustrates the regime switch for the proposed estimator. For large values of $\delta$ (left hand images, notice the scale on the $y$-axis), the proposed estimator behaves similarly to the median of blocks estimator; that is, statistically optimal with high confidence. For regimes with smaller $\delta$ (right hand images, notice the different scale on the $y$-axis), the proposed estimator manages to obtain an accuracy similar to the plug-in estimator.}\label{fig:regimepareto}
\end{figure}

\begin{figure}[H]
	\begin{minipage}{0.5\textwidth}
		\includegraphics[width=1\textwidth]{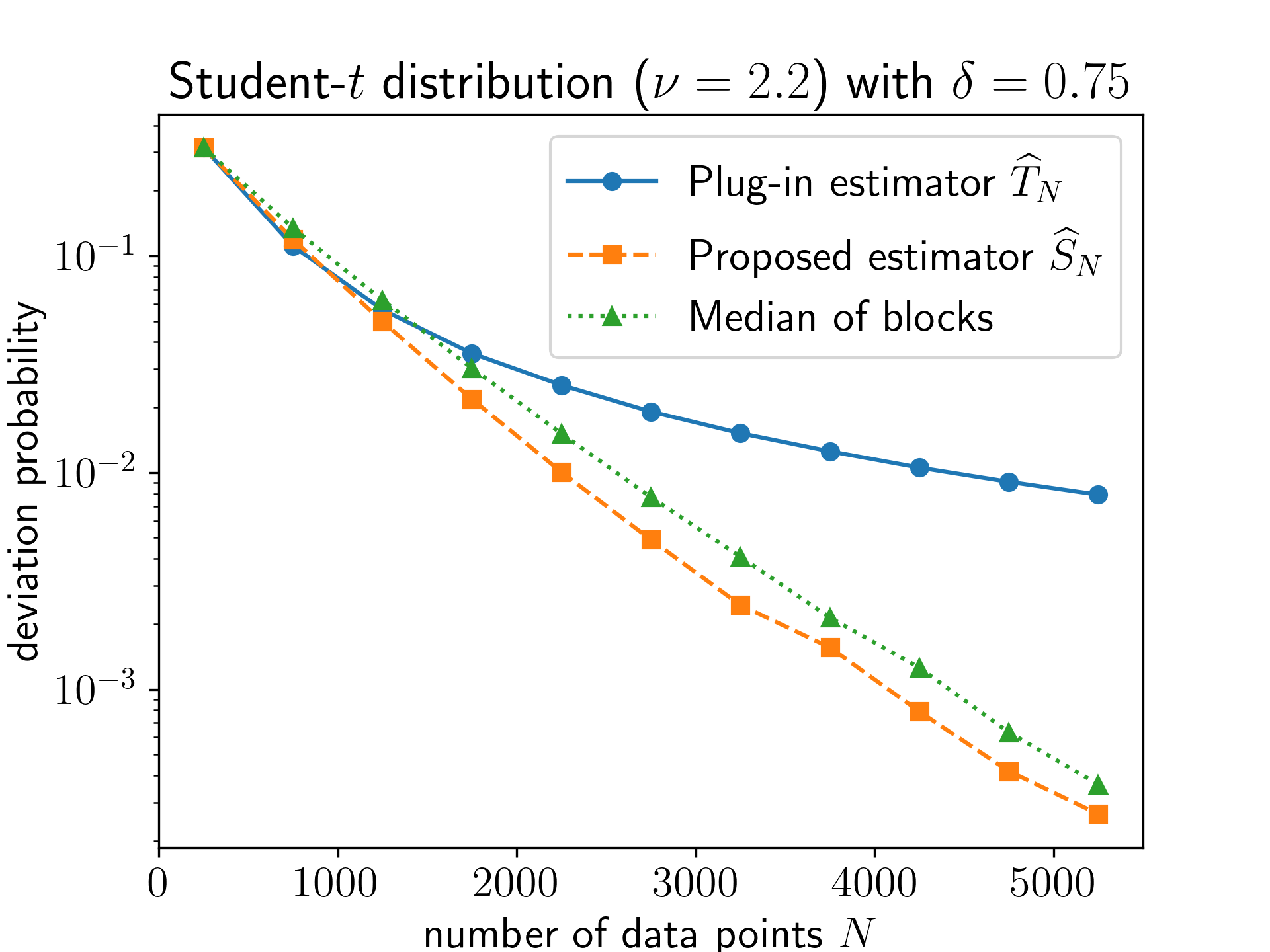}
	\end{minipage}%
	\begin{minipage}{0.5\textwidth}
		\includegraphics[width=1\textwidth]{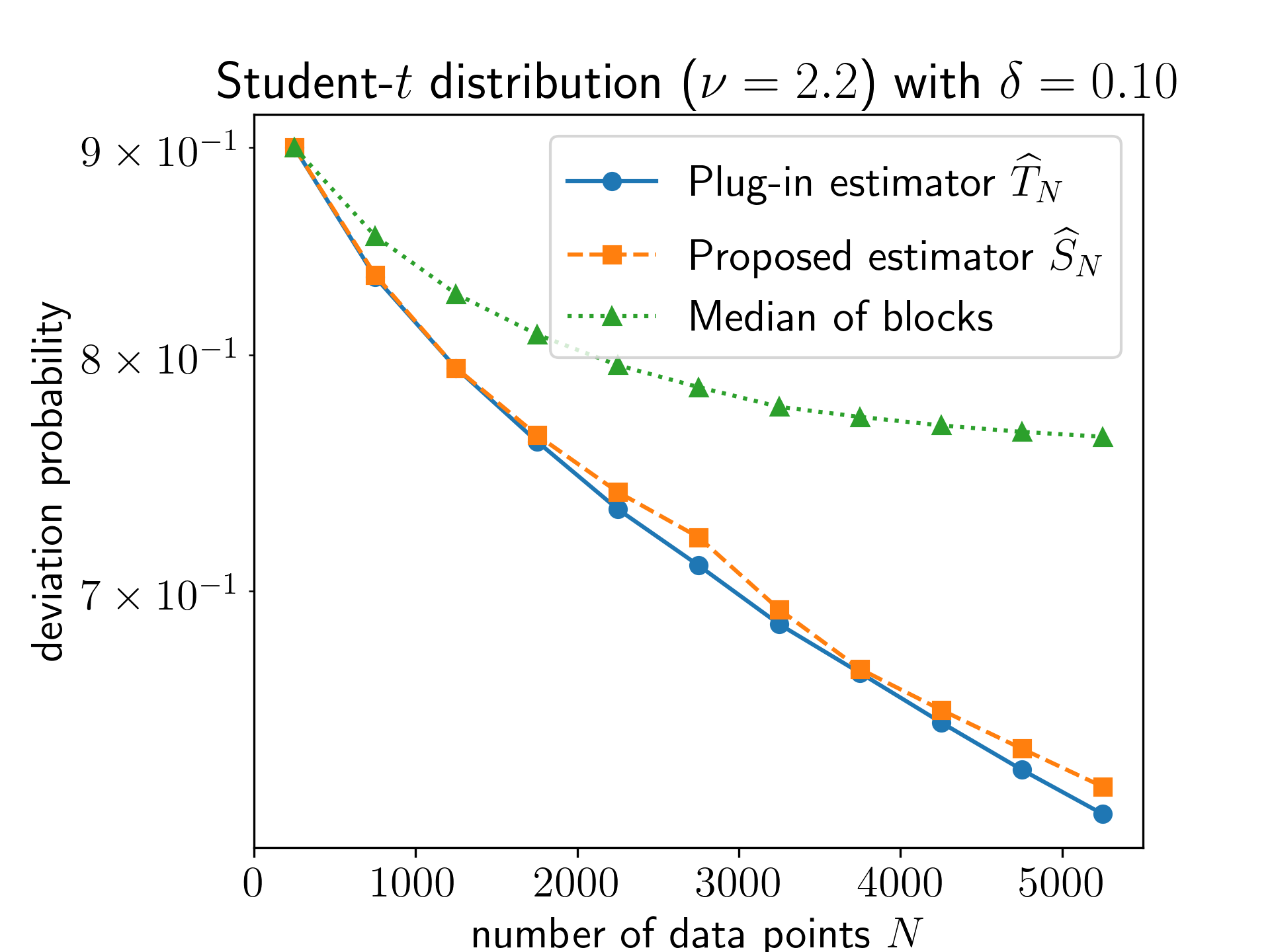}
	\end{minipage}
	\caption{The figures show $\mathbb{P}(|\widehat{R}_N - \ES_{\alpha}(X)| \geq \delta)$ (estimated using $10^6$ many experiments) for different estimators $\widehat{R}_N$, where $X$ follows a student-$t$ distribution with parameter $\nu=2.2$. The proposed estimator and the median of blocks estimator use sub-intervals of size $m=250$.}\label{fig:regimestudent}
\end{figure}

{\color{blue}
\begin{figure}[H]
	\begin{minipage}{0.5\textwidth}
		\includegraphics[width=1\textwidth]{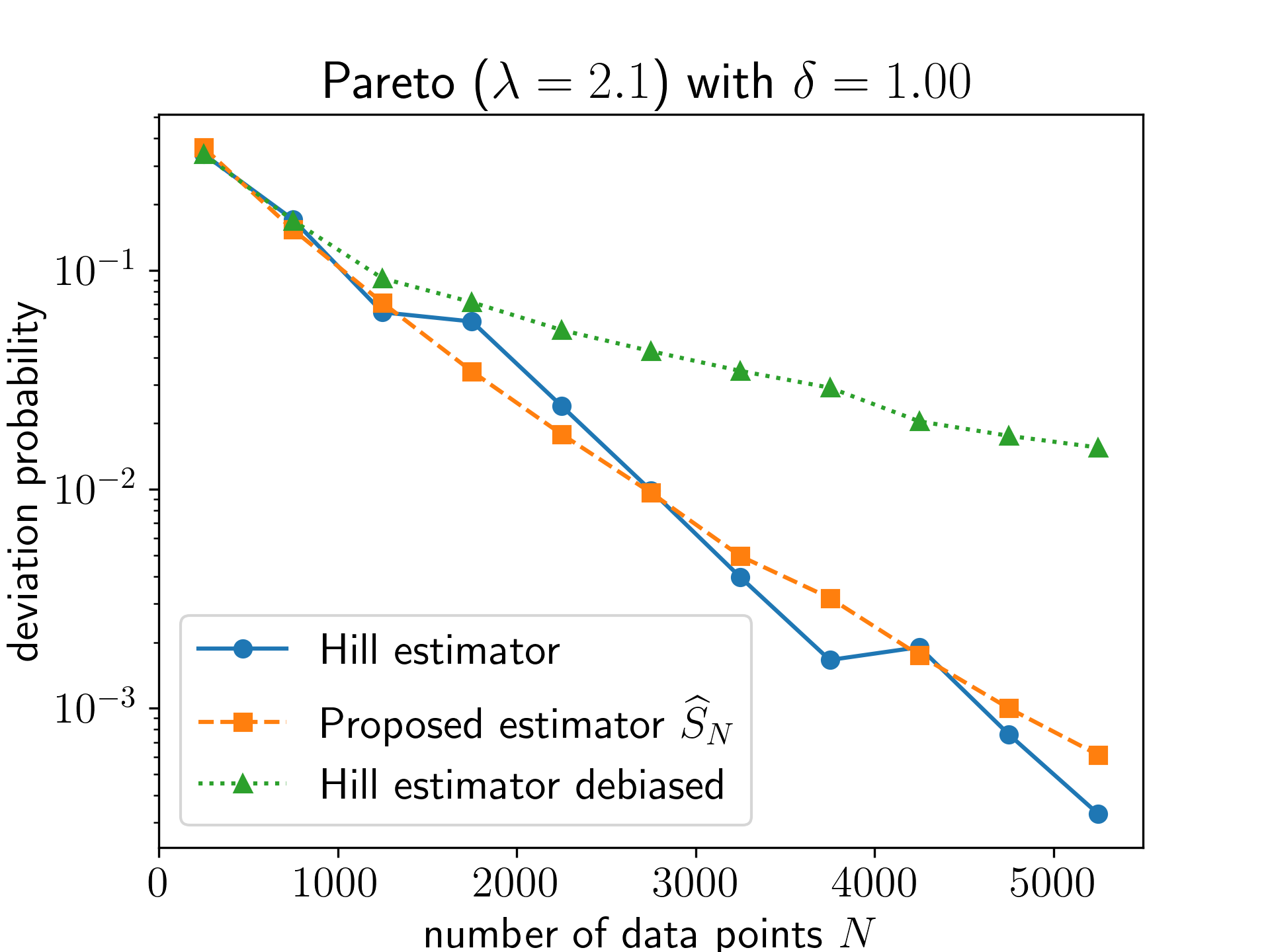}
	\end{minipage}%
	\begin{minipage}{0.5\textwidth}
		\includegraphics[width=1\textwidth]{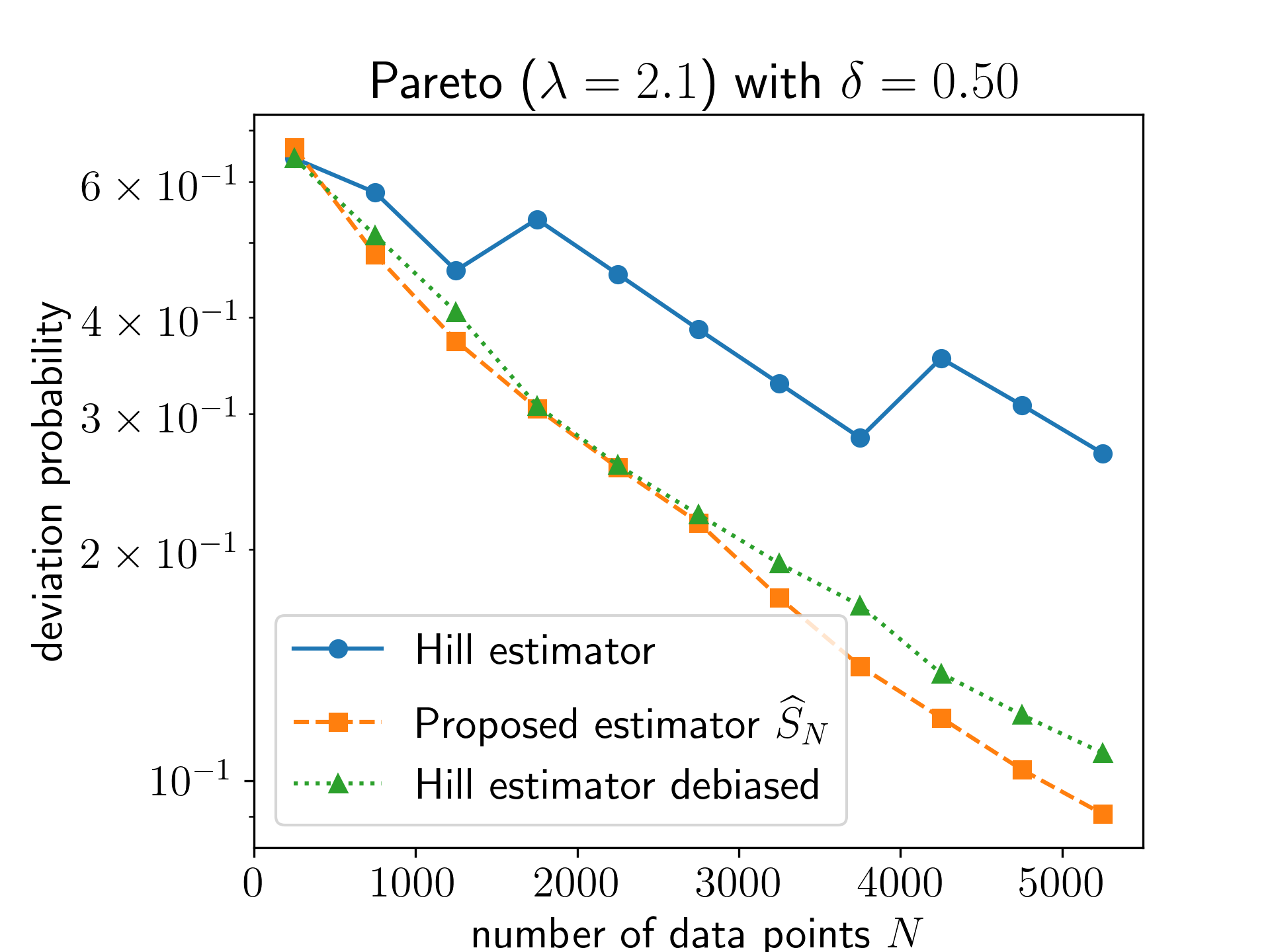}
	\end{minipage}
	\caption{The figures show $\mathbb{P}(|\widehat{R}_N - \ES_{\alpha}(X)| \geq \delta)$ (estimated using $10^6$ many experiments) for different estimators $\widehat{R}_N$, where $X$ follows a Pareto distribution with parameter $\lambda=2.1$. The proposed estimator uses sub-intervals of size $m=250$. Hill's estimator (cf.\cite{hill2015expected}) is a kind of trimmed estimator which cuts off the $[0.25 \cdot n^{1/3}]$ largest data points.}\label{fig:trimmed}
\end{figure}}

\begin{figure}[H]
	\begin{minipage}{0.5\textwidth}
		\includegraphics[width=1\textwidth]{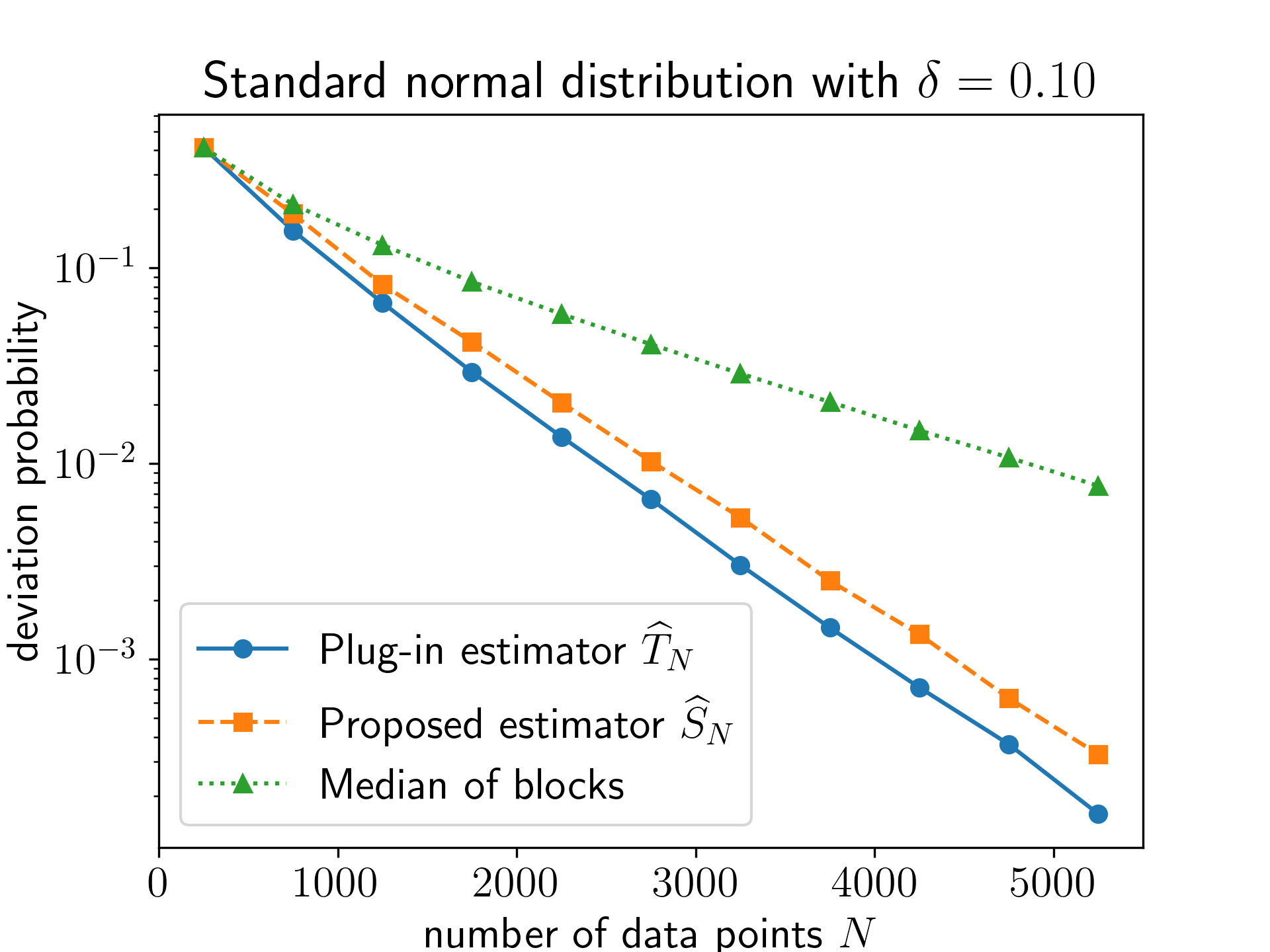}
	\end{minipage}%
	\begin{minipage}{0.5\textwidth}
		\includegraphics[width=1\textwidth]{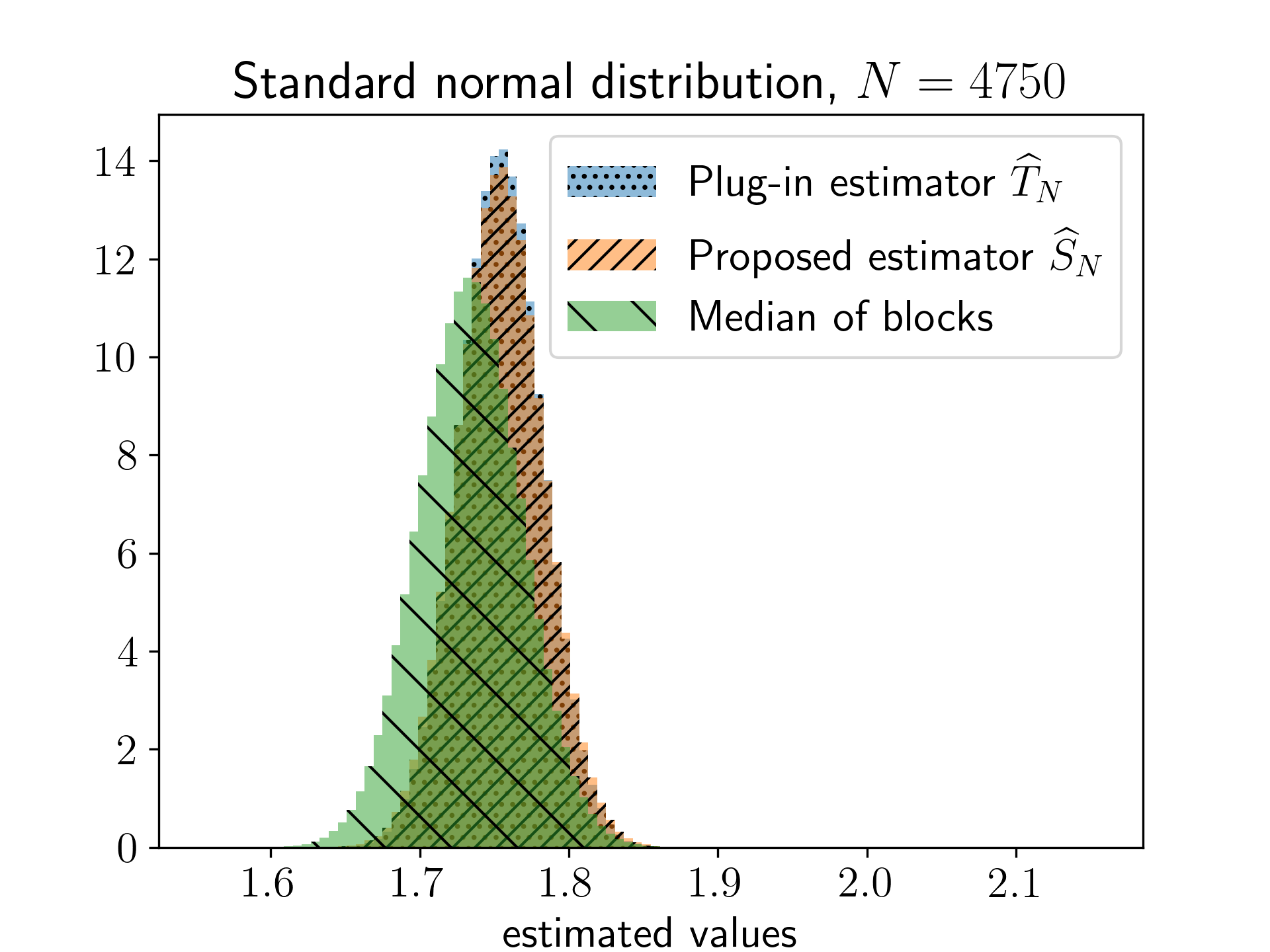}
	\end{minipage}
	\begin{minipage}{0.5\textwidth}
	\includegraphics[width=1\textwidth]{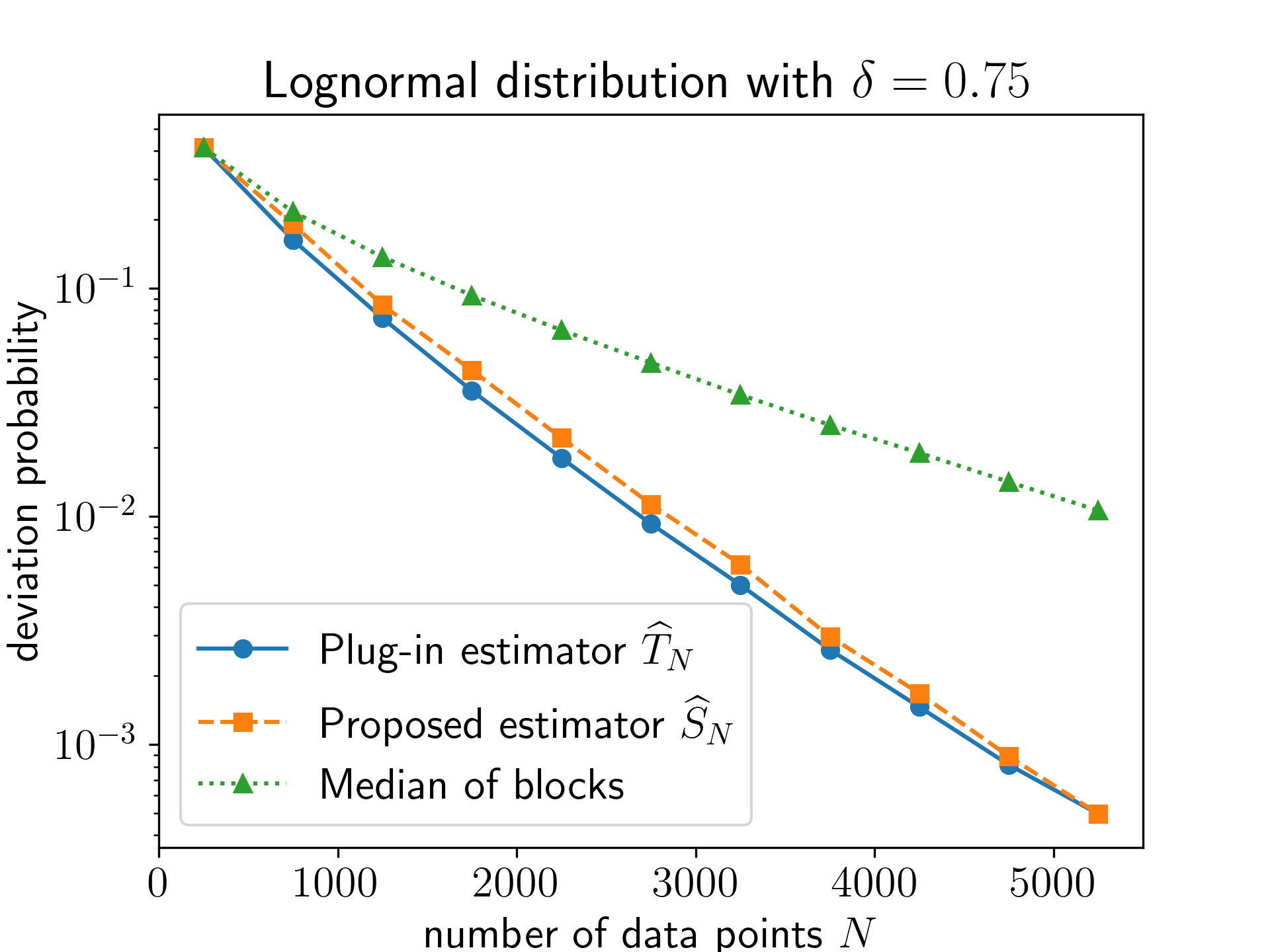}
	\end{minipage}%
	\begin{minipage}{0.5\textwidth}
		\includegraphics[width=1\textwidth]{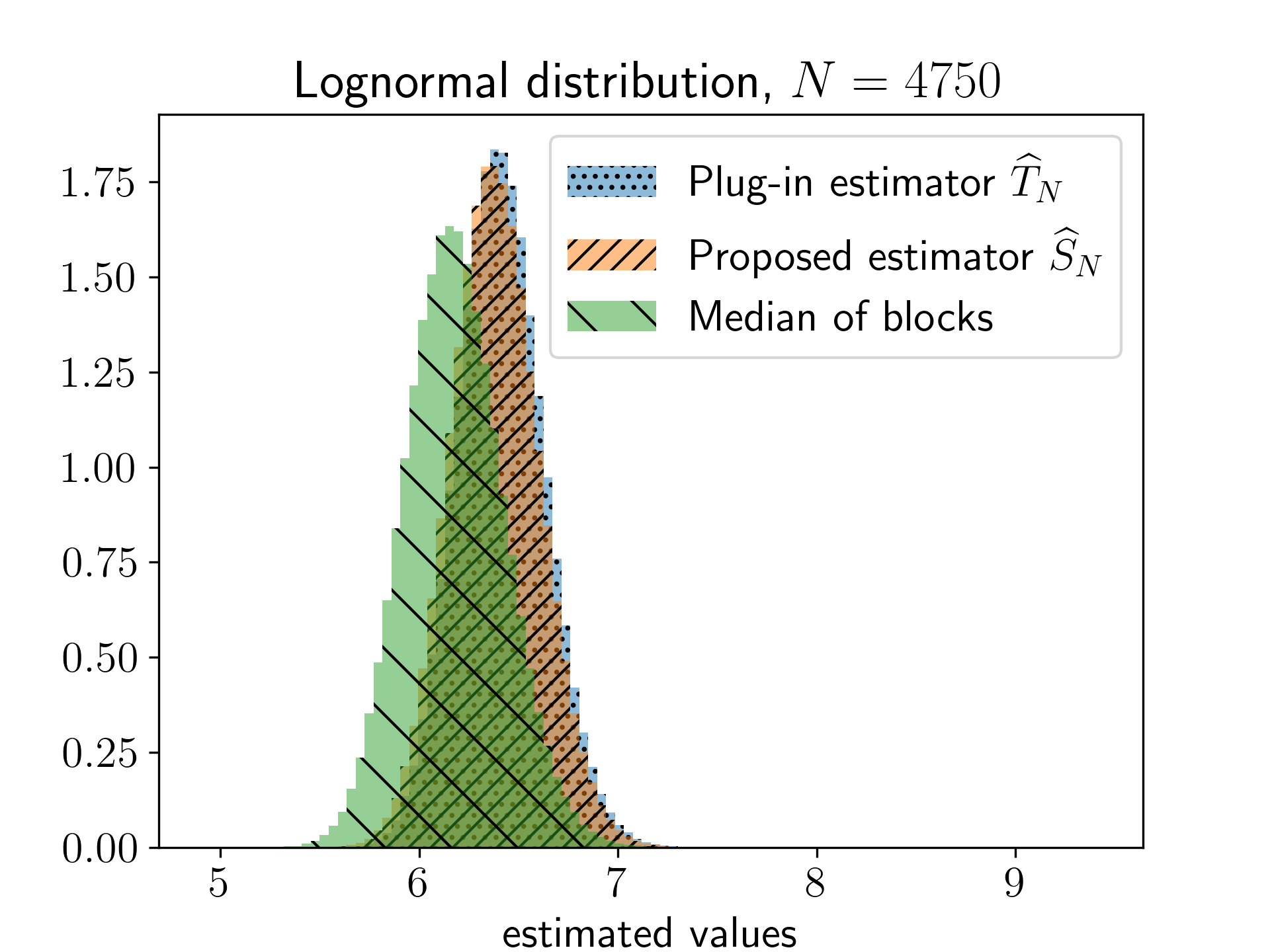}
	\end{minipage}%
	\caption{The left hand side shows $\mathbb{P}(|\widehat{R}_N - \ES_{\alpha}(X)| \geq \delta)$ (estimated using $10^6$ many experiments) for different estimators $\widehat{R}_N$, where $X$ is either standard normally distributed (top) and or standard log-normally distributed (bottom). The proposed estimator and the median of blocks estimator use sub-intervals of size $m=125$. We see that the median of blocks estimator exhibits a noticeable negative bias, as shown by the histograms on the right hand side. The proposed estimator mitigates this downside almost completely and performs similarly to the plug-in estimator in these examples.}\label{fig:nicedistributions}
\end{figure}

\begin{figure}[H]
	\begin{minipage}{0.5\textwidth}
	\includegraphics[width=1\textwidth]{Figures/fullhistpareto22_new.png}
\end{minipage}%
\begin{minipage}{0.5\textwidth}
	\includegraphics[width=1\textwidth]{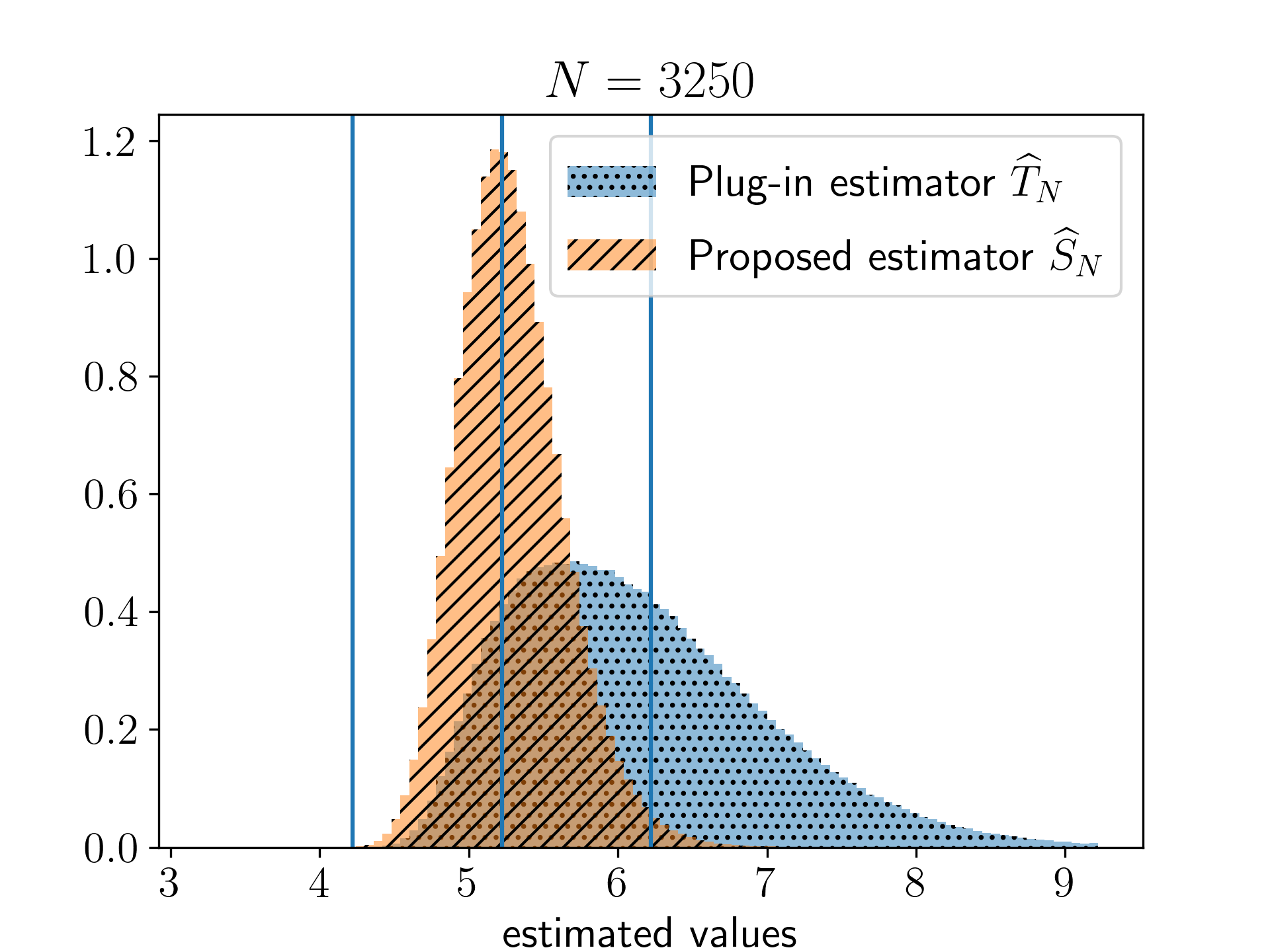}
\end{minipage}
\caption{The figures show histograms across estimated values for $\ES_{\alpha}(X)$, where $X$ is Pareto distributed with $\lambda = 2.2$, with the blue vertical lines depicting the values $\ES_{\alpha}(X)-1, \ES_{\alpha}(X), \ES_{\alpha}(X)+1$. The left hand side shows the standard case without any adversarial manipulation, as reported in the introduction. The right hand side shows the case in which three of the 3250 data points, w.l.o.g.~$X_1, X_2, X_3$, are modified to $\tilde{X}_i = \max\{X_i, U_i\}$ with $U_i \sim \mathcal{N}(5, 250^2)$ independent of all other variables, $i=1, 2, 3$. We see that the plug-in estimator is heavily distorted through those modified data points, while the proposed estimator is much more robust and its histogram remains close to the one using the uncorrupted data.}\label{fig:corrupted}
\end{figure}

\begin{figure}[H]
	\begin{minipage}{0.5\textwidth}
		\includegraphics[width=1\textwidth]{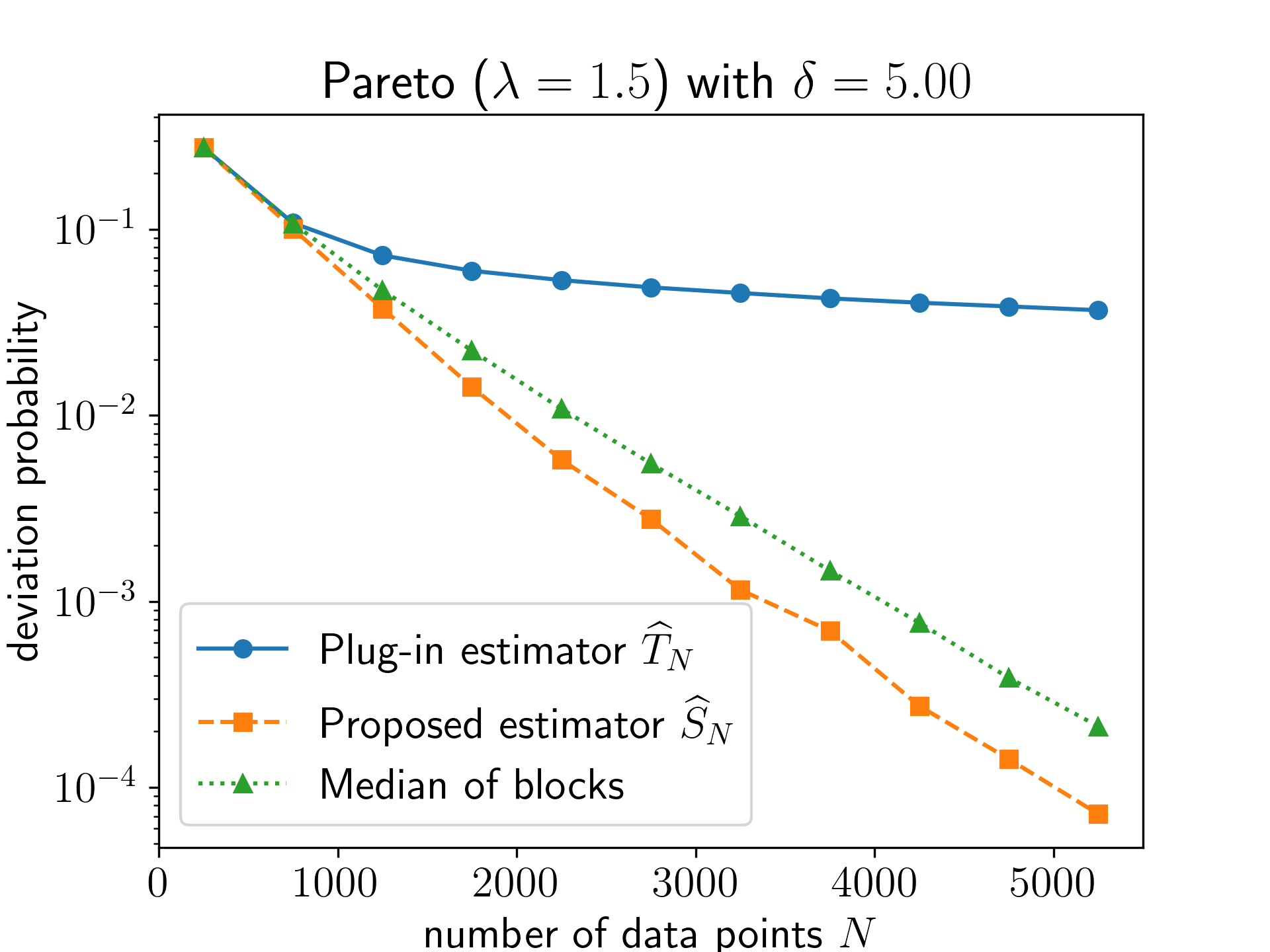}
	\end{minipage}%
	\begin{minipage}{0.5\textwidth}
		\includegraphics[width=1\textwidth]{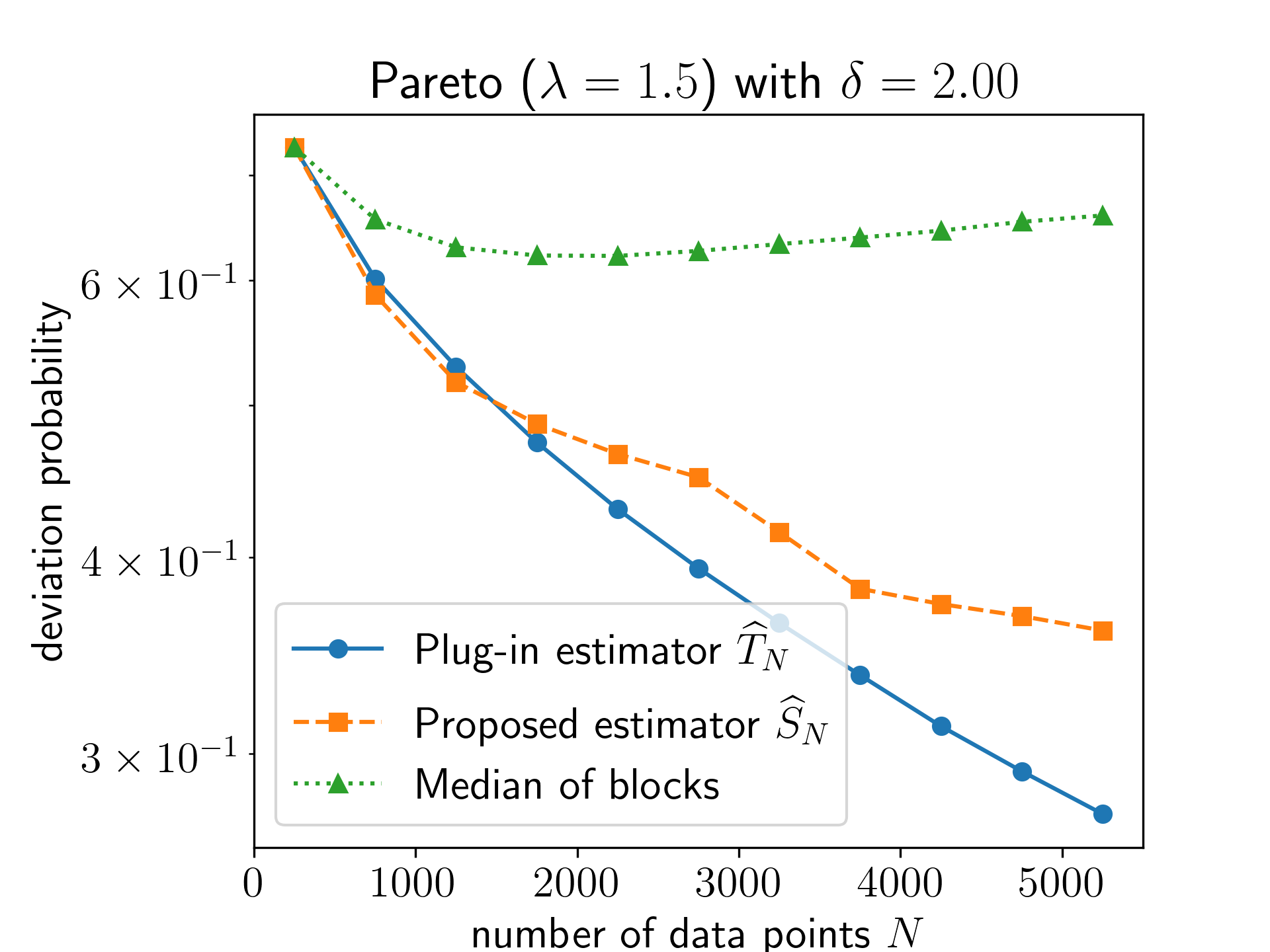}
	\end{minipage}
	\caption{The figures show $\mathbb{P}(|\widehat{R}_N - \ES_{\alpha}(X)| \geq \delta)$ (estimated using $10^6$ many experiments) for different estimators $\widehat{R}_N$, where $X$ is Pareto distributed with parameter $\lambda = 1.5$. The proposed estimator and the median of blocks estimator use sub-intervals of size $m=250$.}\label{fig:theorydoesnothold}
\end{figure}

\section{Extension to non i.i.d.\ data}
\label{sec:non.iid}

The results in this paper are derived under the assumption of (possibly corrupted) i.i.d.\ data. 
While the i.i.d.\ assumption is often used in the academic literature, it may not hold in many practical applications. 
We believe that our main results can be extended to a variety of non-i.i.d.\ settings---we refer to \cite{Cai2008,chen2008nonparametric,mcneil2015quantitative,patton2019dynamic,Tsay2010} for examples of such frameworks. 
Because each dependence setting calls for distinct techniques and assumptions, a full treatment of all extensions lies beyond the scope of this paper. 
In what follows, we outline one plausible extension of our methodology, leaving a comprehensive analysis of further settings for future work.

To that end, recall that for two $\sigma$-fields $\mathcal{A}$ and $\mathcal{B}$,
\[\beta(\mathcal{A},\mathcal{B}):= \sup_{C\in\mathcal{A}\otimes\mathcal{B}} |\P_{\mathcal{A}\otimes\mathcal{B}}(C) - (\P_\mathcal{A}\otimes\P_\mathcal{B})(C)|.\]
Here $\P_\mathcal{A}$ and $\P_\mathcal{B}$ are the restriction of $\P$ to the respective $\sigma$-fields and $\P_{\mathcal{A}\otimes\mathcal{B}}$ is defined via $\P_{\mathcal{A}\otimes\mathcal{B}}(A\times B ) = \P(A\cap B)$.

We assume that $X_1,\dots,X_N$ are observations from a strictly stationary time series which is geometrically $\beta$-mixing, that is 
\[
\beta(k)
= \sup_{t\geq 1 } \beta\left( \sigma(X_1,\dots,X_t), \sigma(X_{t+k},\dots, X_N) \right)
\]
satisfies for all $k\geq 1$,
\[\beta(k) \le C_{{\rm mix}}\rho^{k} \quad\text{for some constants }  C_{{\rm mix}} > 0 \text{ and } \rho \in (0,1).\]
Geometric $\beta$-mixing is satisfied for a wide range of standard financial time-series models used in ES estimation, including ${\rm GARCH}$ and stochastic volatility models under mild moment and contraction assumptions.
It also holds for many ergodic Markov chains and linear processes with exponentially decaying coefficients.

Define for each $t,s\in\R$ and $k \geq 0$,
\[ \gamma_{k}(t,s) = \C\bigl[\,\eins_{\{X_{0} \le t\}},\, \eins_{\{X_{k} \le s\}}\,\bigr],\]
set 
\[\Gamma(t,s) := \gamma_{0}(t,s) + 2\sum_{k \ge 1} \gamma_{k}(t,s),
\]
and put
\[
\overline{\sigma}_{\mathrm{ES}}^{2}
:= \int_{F^{-1}(1-\alpha)}^\infty \int_{F^{-1}(1-\alpha)}^\infty
\Gamma(t,s)\, ds\,dt.
\]
Note that 
\[\gamma_0(t,s) =  F(t \wedge s) - F(t)F(s)\] 
and hence
\begin{align}
\label{eq:non.iid.sigma}
\overline{\sigma}_{\mathrm{ES}}^{2}
= \sigma_{\mathrm{ES}}^{2}
+ 2\sum_{k \ge 1}
\int_{F^{-1}(1-\alpha)}^\infty \int_{F^{-1}(1-\alpha)}^\infty
\gamma_{k}(t,s)\, ds\, dt.
\end{align}
In particular, in the i.i.d.\ case (corresponding to $\rho=0$) we have  $\overline{\sigma}_{\ES_\alpha}^{2} = \sigma_{\ES_\alpha}^{2}$. 
In general, $\overline{\sigma}_{\ES_\alpha}^{2}$ may exceed $\sigma_{\ES_\alpha}^{2}$,
and there are several problem-specific approaches to estimate it. 
For instance, under the strong mixing assumption one can verify that for all $k \ge 1$,
\[
|\gamma_{k}(t,s)|
\le \min\{\, C_{{\rm mix}} \rho^{k},\, 1 - F(t),\, 1 - F(s)\,\}.
\]
Hence,  the remainder term in~\eqref{eq:non.iid.sigma} can be estimated for example if  $F$ possesses suitable higher moments.

The following is the extension of Theorem \ref{thm:ES} to the present setting of non i.i.d.\ data.
For cleaner exposition, we shall assume that $\alpha\geq 0.01$, that $\rho<0.99$, and that $C_{\rm mix} \geq 4$ (ensuring that $\log(C_{\rm mix})\geq 1$).

\begin{theorem}
\label{thm:ES.noniid}
	There are absolute constants $c_0,\dots, c_4>0$ such that the following holds.
	Assume that $\overline{\sigma}^2_{\ES_\alpha(X)}$ is finite and that $u\mapsto \VAR_u(X)$ is $L$-Lipschitz continuous on $[1-2\alpha , 1- \frac{1}{2}\alpha]$.
	Let $\varepsilon< c_0\sqrt\alpha$, assume that $N \geq c_1\log(C_{\rm mix}) \varepsilon^{-2}$ and $N\geq c_{2} \varepsilon^{-4}$.
    Then there exists an estimator $\widetilde{S}_N$ such that
	\[ \P\left( \left| \widetilde{S}_N-\mathrm{ES}_\alpha(X) \right| 
	\geq   \varepsilon \overline{\sigma}_{\ES_\alpha(X)}  + c_3 L \varepsilon^2   \right)
	 \leq 2 \exp\left(-c_4 N\varepsilon^2\right).\]
	\end{theorem}

The estimator $\widetilde{S}_N$ is defined very similarly to the estimator $\widehat{S}_N$ from Definition \ref{def:proposed} with one difference due to the non i.i.d.\ setting. We explain its definition and the differences in Subsection \ref{subsec:noniidproof}.

Note that the error estimate and the confidence in Theorem \ref{thm:ES.noniid} are the same as those established in Theorem~\ref{thm:ES} (up to the exact value of the absolute constants).
The only changes are that the conditions on $N$ have been slightly modified:
The lower bound on $N$ is now $N \gtrsim \log(C_{\rm mix}) \varepsilon^{-2}$ compared to the previous bound $N \gtrsim  \varepsilon^{-2}$. 
As the value of $C_{\rm mix}$ is typically moderate (of order $10$), the new lower bound on $N$ almost matches the previous one.
Second, the new condition on $N$, namely that $N \le c_{2}\varepsilon^{-4}$, is quite natural. 
Indeed, one is typically interested in values of $\varepsilon$ only slightly larger than $N^{-1/2}$; those choices 
yield an estimation error slightly larger than $O(\overline{\sigma}_{\ES_{\alpha}(X)} / \sqrt{N})$. 
In contrast, the upper bound $N \sim \varepsilon^{-4}$, or equivalently, $\varepsilon \sim N^{-1/4}$, already corresponds to a very large error of order $O(\overline{\sigma}_{\ES_{\alpha}(X)} / N^{1/4})$.

Finally note that as previously in Theorem \ref{thm:adversarial}, minor modifications of the proof of Theorem \ref{thm:ES.noniid} yield the following in the setting of adversarially modified non-i.i.d.\ samples.

\begin{theorem}\label{thm:adversarial.noniid}
	Let the assumptions in Theorem \ref{thm:ES.noniid} hold.
	Then there are absolute constants $C_0,C_1,C_2$ such that even if at most $C_0 N\varepsilon^2$ of the $N$ samples are maliciously modified,  
	\[ \P\left( \left| \widetilde{S}_N-\mathrm{ES}_\alpha(X) \right| \geq   \varepsilon\sigma_{\ES_\alpha(X)}  + C_1 L \varepsilon^2   \right)
	\leq 2\exp\left(-C_2 N\varepsilon^2\right).\]
\end{theorem}

\subsection{Proof of Theorems~\ref{thm:ES.noniid} and \ref{thm:adversarial.noniid}}\label{subsec:noniidproof}

We start with the proof of Theorem~\ref{thm:ES.noniid}.
To that end, let us first explain the definition of the estimator $\widetilde{S}_N$.
Let $C_0$ be a well-chosen absolute constant\footnote{One may take $C_0=150C_1$ where $C_1$ is the absolute constant that appears in the proof of Lemma \ref{lem:error}, thought we did not try to optimize constants here and this choice can likely be improved.}, set
\[m=\left\lceil  \frac{C_0  \log(C_{\rm mix})}{\varepsilon^2}\right\rceil, \qquad n:=\left\lfloor \frac{N}{2m}\right\rfloor,\]
and partition the data $X_1,\dots,X_N$ in $2n$ disjoint blocks $I'_1,\dots,I'_{2n}$ containing $m$ samples each. 
Set 
\[I_j:=I'_{2j} \quad\text{for each }j=1,\dots, n.\]
The estimator $\widetilde{S}_N$ is now defined as in Definition \ref{def:proposed} using this choice of the $I_j$'s.

The reason for this choice of the $I_j$'s is the following (important) observation: the mixing coefficient of the sequence $(\widehat{T}_{I_{j}})_{j=1}^n$, 
denoted by $\beta^{T}$, satisfies
\[
\beta^{T}(k) \le \beta(km) \le C_{{\rm mix}}\rho^{mk},
\]
where $m$ denotes the size of each block $I_{j}$.
This follows from the definition of the mixing coefficients and the fact that the gap in samples $X_1,\dots,X_N$ between $\widehat{T}_{I_{j}}$ and $\widehat{T}_{I_{k}}$ is $|j-k|m$, e.g.\ $\widehat{T}_{I_1}$ uses the samples $X_1,\dots,X_m$ and $\widehat{T}_{I_2}$ uses the samples $X_{2m+1},\dots,X_{3m}$ and so forth.

In what follows (in Lemma \ref{lem:linear.noniid} and Lemma \ref{lem:error.noniid} below), we will show that for each $j$,
\begin{align}
\label{eq:non.iid.to.show}
\P\left( \left|\widehat{T}_{I_{j}} - \ES_{\alpha}(X)\right| \geq \delta \right) \leq \frac{1}{20},
\end{align}
where here and in the reminder of this section, 
\[ \delta = \varepsilon \overline{\sigma}_{\ES_\alpha(X)}  + 10 L \varepsilon^2.\]
If that is true, the following analogue to Lemma \ref{lem:binomial.concentration} in the non-i.i.d.\ setting is used to obtain high probability estimates:

\begin{lemma}
\label{lem:binomial.noniid}
Assume that the condition on $N$ from Theorem \ref{thm:ES.noniid} is satisfied and suppose that \eqref{eq:non.iid.to.show} holds true for every $j=1,\dots,n$.
Then,
\begin{align*}
\P\left( \left|\{ j\in\{1,\dots,n\} : |\widehat{T}_{I_{j}} -  \ES_{\alpha}(X)| \geq  \delta \}\right| \geq 0.3n \right) 
\leq 2  \exp\left(-\frac{n}{50}\right).
\end{align*}

\end{lemma}

The proof of Lemma \ref{lem:binomial.noniid} is based on a coupling argument.
Though we suspect that the following statement is known, we could not find it exactly  in the literature, which is why we include a proof.

\begin{lemma}
\label{lem:coupling}
    Suppose that $(Z_t)_{t\geq 1}$ is $\beta$-mixing.
    Then one can construct (on some possible enlarged probability space) a new sequence $(Y_t)_{t\geq 1}$ such that:
    \begin{enumerate}
        \item $Y_t$ and $Z_t$ have the same distribution for each $t$,
        \item the sequence $(Y_t)_{t\geq 1}$ is independent,
        \item $\P(Y_t\neq Z_t)\leq\beta(1)$ for all $t\geq 1$.
    \end{enumerate}
\end{lemma}
\begin{proof}
The proof proceeds by induction over $T$ for the sequence $(Z_t)_{t=1}^T$.
We additionally claim that each $Y_t$ is constructed from $Z_t$ together with a uniform random variable $U_t$, where the $U_t$'s are i.i.d.\ and independent of everything else.
For $T=1$, we simply set $Y_1 := Z_1$. Thus, suppose that the claim holds for some $T \ge 1$, and consider the step from $T$ to $T+1$.
Define
\[
\mathcal{A} := \sigma(Z_T)
\qquad \text{and} \qquad
\mathcal{B} := \sigma(Z_1,\dots,Z_{T-1},\,Y_1,\dots,Y_{T-1}).
\]
Then
\[
\beta(\mathcal{A},\mathcal{B})
\le \beta(\mathcal{A},\sigma(Z_1,\dots,Z_{T-1}))
\le \beta(1).
\]
Indeed, the second inequality follows directly from the definition of $\beta$, since
\[
\mathcal{B} \subset \sigma(Z_1,\dots,Z_{T-1},U_1,\dots,U_{T-1}),
\]
and $(U_t)_{t\ge1}$ is independent of $(Z_t)_{t\ge1}$.

Therefore, by Berbee’s lemma (see, e.g., Lemma~5.1 in \cite{Rio2017}), there exists a random variable $Y_T$, independent of $\mathcal{B}$, such that
\[
Y_T \sim Z_T,
\qquad
\mathbb{P}(Y_T \neq Z_T) \le \beta(1),
\]
and moreover $Y_T$ is measurable with respect to $\sigma(Z_T,U_T)$.

This completes the induction step and hence the proof.
\end{proof}

\begin{proof}[Proof of Lemma \ref{lem:binomial.noniid}]
By Lemma \ref{lem:coupling}, one can define independent random variables $(Y_{j})_{j=1}^{n}$ such that 
\begin{align*}
\P(Y_j \neq \widehat{T}_{I_{j}} ) \leq \beta^T(1) \quad\text{for every }j=1,\dots,n.
\end{align*}
Since $\beta^T(1) \leq C_{\rm mix} \rho^{m}$, we obtain
\begin{align}
\label{eq:berbee.2}
\begin{split}
 \P( Y_j \neq \widehat{T}_{I_{j}} \text{ for some } j\leq n ) 
&\leq n C_{{\rm mix}} \rho^{m}.
\end{split}
\end{align}
Next note that 
\[ n C_{{\rm mix}}  \rho^m \leq  \exp\left( -\frac{n}{50}\right)  \]
holds whenever
\[ m\geq \frac{n/50 + \log(C_{\rm mix}) + \log(n) }{\log(1/\rho)}. \]
Since $\rho\leq 0.99$, the latter is satisfied whenever $m\geq C_1 n \log(C_{\rm mix})$ for a suitable absolute constant $C_1$.
By the definition of $n$ this inequality is satisfied if $N\leq C_2\varepsilon^{-4}$ which we may assume.

In particular, \eqref{eq:berbee.2} shows that $\P( Y_j \neq \widehat{T}_{I_{j}}) \leq\frac{1}{20}$ for all $j$ since by assumption $N\varepsilon^2 \geq  C_3 \log(C_{\rm mix})$ for a suitable large absolute constant $C_3$ and thus $\exp(-n/50) \leq\frac{1}{20}$.
Therefore, 
\[ \P( \left|Y_{j} -  \ES_{\alpha}(X)\right| \geq \delta  )
\leq \P( |\widehat{T}_{I_{j}} -  \ES_{\alpha}(X)| \geq \delta  ) + \P(Y_j \neq \widehat{T}_{I_{j}} )
\leq \frac{1}{10}.\]
and hence an application of Lemma \ref{lem:binomial.concentration} shows that
\begin{align}
\label{eq:berbee.3}
\P( |\{ j\in\{1,\dots,n\} : |Y_j -  \ES_{\alpha}(X)| \geq \varepsilon \delta \}| \geq 0.3n ) 
\leq   \exp\left(-\frac{n}{50}\right).
\end{align}
The proof of the lemma  follows from combining \eqref{eq:berbee.2} with \eqref{eq:berbee.3} and using the union bound.
\end{proof}

Thus, it remains to show \eqref{eq:non.iid.to.show}.
To that end we follow the same reasoning as in Section \ref{sec:mainproofs}, and---using the same notation as therein---write
\[\widehat{T}_{I_{j}} - \ES_{\alpha}(X) = \mathcal{L}_j + \mathcal{E}_{j}.\]

\begin{lemma}
\label{lem:linear.noniid}
	For every $j=1,\dots,n$,
\[ \P( |\mathcal{L}_{j} |\geq\varepsilon \overline{\sigma}_{\ES_{\alpha}} ) \leq \frac{1}{11} \]
\end{lemma}
\begin{proof}
	The proof follows the same reasoning as the proof of Lemma \ref{lem:linear}; in particular
	\begin{align*}
	\E[ \mathcal{L}_j^2 ]
	=\frac{1}{\alpha^2} \int_{F^{-1}(1-\alpha)}^\infty \int_{F^{-1}(1-\alpha)}^\infty  \C\left[\widehat{F}_{I_j}(t), \widehat{F}_{I_j}(s) \right] \, dt ds.
	\end{align*}
	Using the definition of $\gamma_{k}$, it follows that
	\begin{align}
	\label{eq:non.iid.cov}
    \nonumber
	\C\left[ \widehat{F}_{I_{j}}(t),  \widehat{F}_{I_{j}}(s)   \right]
	&=  \frac{1}{m^2} \sum_{i=1}^m\sum_{\ell=1}^m \C\left[ \eins_{X_i\leq t}, \eins_{X_\ell\leq s}   \right] \\
    \nonumber 
    &=\frac{1}{m^2}\sum_{i=1}^m\sum_{\ell=1}^m \gamma_{|i-\ell|}(t,s) \\
    \nonumber 
&=\frac{1}{m}\sum_{k=-(m-1)}^{m-1}\left(1-\frac{|k|}{m}\right)\gamma_{|k|}(t,s)  \\
	&\leq  \frac{1}{m}  \Gamma(t,s),
	\end{align}
	where the last inequality holds by the definition of $\Gamma$.
	Thus $\E[ \mathcal{L}_j^2 ] \leq \frac{1}{m} \overline{\sigma}^{2}_{\ES_{\alpha}}$, and the proof is completed by an application of Chebychev's inequality and since $m\geq 11/\varepsilon^2$.
\end{proof}

\begin{lemma}
\label{lem:error.noniid}
Set $\gamma:= \frac{5}{4} \varepsilon\sqrt\alpha$.
	Then, for every $j=1,\dots,n$,
\[ \P\left(|\mathcal{E}_j|>  \frac{ 3\gamma  }{\alpha} \left( F^{-1}(1-\alpha+\gamma) -F^{-1}(1-\alpha- \gamma) \right)\right)
	\leq \frac{1}{110}.\]
\end{lemma}
\begin{proof}
	The proof follows from the same strategy as in Lemma \ref{lem:error}.
	In fact, only Step 1 therein changes, all other steps follow verbatim the same. 
	
	Recall that $t_{0}$ and $t_{1}$ were chosen such that $F(t_{0}) = 1- \alpha-\gamma$ and  $F(t_{1}) = 1- \alpha+\gamma$.
	We proceed to show that  for $t=t_{0},t_{1}$,
	\begin{align}
	    \label{eq:mix.to.show}
        \P\left(| \widehat{F}_{I_{j}}(t) - F(t)|\geq \gamma \right)\leq \frac{1}{220} 
	\end{align} 
	If that is true, the proof follows from the same reasoning as in Lemma \ref{lem:error}.
	
	Fix $t\in\{t_{0},t_{1}\}$.	
    Using stationarity and the Cauchy-Schwartz inequality
    \begin{align*}
        \gamma_k(t,t)
        &\leq \sqrt{ \V[\eins_{X_1\leq t}] \V[\eins_{X_{k+1}\leq t}]} \\
        &= \V[\eins_{X_1\leq t}] \\
        &=F(t)(1-F(t))
        \leq 1- F(t).
        \end{align*}
    On the other hand, 
    \begin{align*}
        \gamma_k(t,t) 
        &= \P(X_1\leq t, X_{k+1}\leq t)  -\P(X_1\leq t)\P(X_{k+1}\leq t) \\
        &\leq \beta(\sigma(X_1),\sigma(X_{k+1})) 
        \leq \beta(k)
    \end{align*}
    and therefore
	\[\gamma_{k}(t,t) \leq \min\{ C_{{\rm mix}} \rho^{k}, 1-F(t) \}.\]
	By \eqref{eq:non.iid.cov}, this entails that
	\[ \V\left[ \widehat{F}_{I_{j}}(t) \right]
	\leq   \frac{1}{m} \sum_{k\geq 0 } \min\{ C_{{\rm mix}} \rho^{k}, 1-F(t) \}. \]
    Note that $1-F(t)\in[\frac{1}{2}\alpha,2\alpha]$ and thus 
    \[C_{{\rm mix}} \rho^{k} \leq 1-F(t)\quad\text{for }k\geq k_0:= \frac{\log(2C_{\rm mix}/\alpha) }{\log(1/\rho)}.\]
    Therefore,
	\begin{align*} 
    \V\left[ \widehat{F}_{I_{j}}(t)\right]
	&\leq  \frac{1}{m}\left( k_0 2 \alpha +  C_{\rm mix}\sum_{k\geq k_0} \rho^k \right) \\
    &=\frac{1}{m}\left( k_0 2 \alpha +  C_{\rm mix}\frac{\rho^{k_0}}{1-\rho} \right)  \\
   & \leq \frac{\alpha C_1\log(C_{\rm mix})}{m}
    	\end{align*}
    where $C_1$ is an absolute constant (since $\rho\leq 0.9$ and $\alpha\geq 0.01$).
	Thus, by Chebychev's inequality, 
	\begin{align*}
	\P\left( | \widehat{F}_{I_{j}}(t) - F(t)|\geq \gamma \right) 
	&\leq  \frac{\alpha C_1\log(C_{\rm mix})}{m \gamma^2}.
	\end{align*}
	Finally, by the choice of $\gamma$,
	\begin{align*}
         \frac{\alpha C_1\log(C_{\rm mix})}{m \gamma^2}
         =  \frac{16 \varepsilon^2 C_1\log(C_{\rm mix})}{25 m}
         \leq \frac{1}{220},
	\end{align*}
	where  inequality holds by the condition on $m$.
    This shows \eqref{eq:mix.to.show} and therefore completes the proof. 
\end{proof}

Note that Lemma \ref{lem:linear} combined with Lemma \ref{lem:error} imply \eqref{eq:non.iid.to.show}.

\begin{proof}[Proof of Theorem \ref{thm:ES.noniid}]
The proof follows exactly as the proof of Theorem \ref{thm:ES}, only substituting  Lemma \ref{lem:binomial.concentration} with Lemma \ref{lem:binomial.noniid}, Lemma \ref{lem:linear} with Lemma \ref{lem:linear.noniid}, and Lemma \ref{lem:error} with Lemma \ref{lem:error.noniid}.
\end{proof}

\begin{proof}[Proof of Theorem \ref{thm:adversarial.noniid}]
The proof follows from an identical reasoning as used for the proof of Theorem \ref{thm:adversarial}.
\end{proof}

\vspace{1em}
\noindent
\textsc{Acknowledgements:}
Daniel Bartl is grateful for financial support through the Austrian Science Fund [grant doi: 10.55776/ESP31 and 10.55776/P34743]  and a Presidential-Young-Professorship grant [‘Robust Statistical Learning from Complex Data’]. 
Stephan Eckstein is grateful for support
by the German Research Foundation through Project 553088969 as well as the Cluster
of Excellence “Machine Learning — New Perspectives for Science” (EXC 2064/1 number
390727645).
The authors are grateful to Felix Liebrich and Gilles Stupfler for helpful comments and suggestions.

\vspace{1em}
\noindent
\textsc{Competing interests:}
There are no competing interests to declare.

 \bibliographystyle{abbrv}
 \bibliography{bib_risk_heavy_tails}

\end{document}